\documentclass[onecolumn,12pt,journal,final]{IEEEtran}

\IEEEoverridecommandlockouts
\usepackage{amsfonts}
\usepackage[dvips]{graphicx}
\usepackage{times}
\usepackage{cite}
\usepackage{amsmath}
\usepackage{cases}
\usepackage{array}
\usepackage{dsfont}
\usepackage{amssymb}

\usepackage{stfloats}
\usepackage{graphicx}
\usepackage{epstopdf}
\usepackage{footnote}
\usepackage{booktabs}
\usepackage{multirow}
\usepackage{bm}
\usepackage{empheq}
\usepackage[labelformat=simple]{subcaption}

\usepackage{amsthm}
\usepackage{algorithm}
\usepackage{algorithmic}
\usepackage{color}
\usepackage{bbm}
\usepackage{subeqnarray}
\usepackage{threeparttable}
\usepackage{cleveref}
\usepackage{bm}

\theoremstyle{plain}
\newtheorem{theorem}{Theorem}
\newtheorem{corollary}{Corollary}

\theoremstyle{definition}
\newtheorem{example}{Example}
\newtheorem{definition}{Definition}
\newtheorem{remark}{Remark}

\begin{document}


\title{Content Caching and Delivery \\ in Wireless Radio Access Networks}

\author{Meixia Tao, Deniz G\"{u}nd\"{u}z, Fan Xu, and Joan S. Pujol Roig
\thanks{M. Tao and F. Xu are with the Department of Electronic Engineering at Shanghai Jiao Tong University, Shanghai, 200240, China (Emails: mxtao@sjtu.edu.cn, xxiaof@sjtu.edu.cn). D. G\"{u}nd\"{u}z and J. Pujol Roig are with the Department of Electrical and Electronic Engineering at Imperial College London, UK (Emails: d.gunduz@imperial.ac.uk, j.pujol-roig16@imperial.ac.uk). }
\thanks{The work by Meixia Tao and Fan Xu is supported by the National Natural Science Foundation of China under grant 61571299 and the Science and Technology Commission of Shanghai Municipality under grant 18DZ2270700. Deniz G\"{u}nd\"{u}z received support from European Research Council (ERC) through Starting Grant BEACON (grant agreement no. 677854). S. Pujol Roig acknowledges funding from the Engineering and Physical Sciences Research Council (EPSRC) and Toshiba Research Europe through an iCASE award to carry out his PhD studies.}}

\maketitle

\begin{abstract}

Today's mobile data traffic is dominated by content-oriented traffic. Caching popular contents at the network edge can alleviate network congestion and reduce content delivery latency. This paper provides a comprehensive and unified study of caching and delivery techniques in wireless radio access networks (RANs) with caches at all edge nodes (ENs) and user equipments (UEs).
Three cache-aided RAN architectures are considered: RANs without fronthaul, with dedicated fronthaul, and with wireless fronthaul. It first reviews in a tutorial nature how caching facilitates interference management in these networks by enabling interference cancellation (IC), zero-forcing (ZF), and interference alignment (IA). Then, two new delivery schemes are presented. One is for RANs with dedicated fronthaul, which considers centralized cache placement at the ENs but both centralized and decentralized placement at the UEs. This scheme combines IA, ZF, and IC together with soft-transfer fronthauling. The other is for RANs with wireless fronthaul, which considers decentralized cache placement at all nodes. It leverages the broadcast nature of wireless fronthaul to fetch not only uncached but also cached contents to boost transmission cooperation among the ENs. Numerical results show that both schemes outperform existing results for a wide range of system parameters, thanks to the various caching gains obtained opportunistically.


\end{abstract}


\section{Introduction}
Over the last decade, mobile data traffic has undergone a significant transformation; not only it has continuously grown at an exponential rate, but also it has become dominated by  content oriented traffic rather than the traditional connection-centric traffic. Currently the network data traffic is dominated by requests for multimedia contents, such as on-demand video streaming and push media\cite{liu2014content,cisco}. This type of traffic can be largely characterized by asynchronous requests for pre-recorded contents, e.g., movies or user-generated content. Moreover,  a large percentage of these requests are for a relatively small number of highly popular contents. These characteristics call for \textit{caching} of popular contents closer to the end users, which can help reduce both the traffic over the network and the latency in delivery. The idea of caching has already been successfully implemented in the Internet through the content distribution networks (CDNs). In recent years,  there has been growing research interest towards pushing content caching all the way to the wireless network edge.
%
%
Caching popular contents locally at macro base stations (MBSs), small base stations (SBSs), or even directly at user equipments (UEs) in a radio access network (RAN) during off-peak traffic periods, can help boost the network performance, similarly to the CDNs' role in the Internet.

Caching at the wireless network edge has its own challenges and characteristics that distinguish it from traditional solutions in a wired network. Most distinctively, wireless is a broadcast medium, which leads to interference, but also allows multiple requests to be served simultaneously from the same base station. Similarly, signals from multiple base stations can act as interference, but can also be exploited through advanced signal processing techniques, such as cooperative multi-point (CoMP) transmission and interference alignment (IA), to improve the reception quality. As we will outline in this paper these characteristics can lead to novel caching gains to be exploited in wireless networks.

The existence of potential gains from coded caching and delivery in a broadcast delivery model that go beyond the local gains from classical uncoded caching, is first shown in the seminal work of Maddah-Ali and Niesen \cite{fundamentallimits}. In \cite{fundamentallimits}, the authors consider a server holding a library of files serving multiple cache-enabled users over a shared broadcast link, and show that a \textit{global caching gain} can be obtained, which, unlike local caching gain, scales with the total number of caches in the network, by leveraging a novel file-splitting based cache placement scheme and coded multicast transmissions. It is further shown in \cite{fundamentallimits} that the performance of this coded caching scheme is within a constant gap to the information-theoretic optimum.

While the gains in \cite{fundamentallimits} require carefully coordinating the cache placement across all the users, in \cite{decentralized} the authors extended their work to \textit{decentralized} cache placement, where the users simply cache random bits from the files in the library. Using a coded  delivery scheme similar to the one in \cite{fundamentallimits}, the authors showed that a global caching gain is still possible. The coded caching framework in \cite{fundamentallimits}  is also studied for the system with non-uniform file popularity \cite{nonuniformdemands, nonuniformdemand10, Ozfatura:Nonuniform:18}, in an online caching system \cite{online}, with finite subpacketization \cite{Shanmugam-IT16},  distinct file sizes \cite{zhang2015coded},  heterogeneous cache sizes \cite{wang2015coded, Amiri:TC:17}, and distinct quality requests from users with distinct cache sizes \cite{Yang:IT:18}.  Apart from the shared link model, coded caching is also studied in other types of networks, such as a hierarchical network \cite{hierarchical}, a device-to-device network \cite{ji2016wireless,D2D}, a multi-level caching network \cite{hachem2014multi}, and a multi-server network \cite{multi-server}.

While the aforementioned works are built upon the error-free shared-link model of \cite{fundamentallimits}, a noisy broadcast channel is a more appropriate model for the downlink in a  wireless RAN, bringing the system model one step closer to reality. In \cite{Bidokhti:IT:18} and \cite{Amiri:TC:18}, content delivery over an erasure broadcast channel is considered, while a Gaussian broadcast channel is studied in \cite{Zhang:ISIT:17}, \cite{Bidokhti:ICC:17} and \cite{Amiri:JSAC:18}. Erasure and Gaussian broadcast delivery channels with feedback are studied in \cite{GhorbelErasureCacheFeedback} and \cite{zhang2017fundamental}, respectively. Main challenge in these works is to exploit the broadcast channel in a non-trivial manner, that goes beyond reducing the problem to delivery over a shared link whose rate is dictated by the user with the worst channel quality. A common conclusion of these works is that, caches at the UEs can compensate for weaker channel conditions.

Content delivery from a single server to multiple users does not reflect the full complexity of wireless RANs we have today. With increasing network densification, users are typically within the coverage area of multiple SBSs, called the \textit{edge nodes} (ENs), which can cooperate to deliver requests to multiple users. When the coded caching framework is extended to a wireless RAN with multiple ENs, several new and interesting research challenges emerge, which will be the focus of this paper. First of all, in a cache-aided RAN, one can consider caches at the ENs as well as caches at the UEs. In the presence of a single broadcasting server, it is natural to assume that the server has access to all the files in the library that can be requested by the users, whereas in the presence of multiple cache-aided ENs, each EN can hold a portion of the library. Cache placement at the ENs and the associated delivery techniques lead to many interesting and challenging problems. A $3\times3$ interference network with caches only at the transmitters (i.e., ENs) is studied in \cite{upperbound}. The authors propose a caching scheme that transforms the original interference network into a combination of broadcast channels, X channels, or hybrid channels, depending on how the subfiles are stored across the transmitters. Then, they use zero-forcing (ZF) and IA techniques in the delivery phase to exploit the presence of the same portions of the files at multiple transmitters. The authors in \cite{lowerbound} introduce the \emph{normalized delivery time} (NDT) as a performance metric in cache-aided interference networks, and  present a lower bound on the NDT in a network with caches only available at the transmitter side. They show that the scheme in \cite{upperbound} is optimal in certain transmitter cache size regions. Note that a similar latency-oriented performance metric is also considered in \cite{zhang2017fundamental}. The model in \cite{upperbound} is later extended to the more general $K_T\times K_R$ cache-aided interference network with caches at both the transmitter and receiver sides in \cite{bothcache,niesen,mine,gunduz,cao}. It is worth mentioning that caches at the UEs have yet another benefit in this context, as the locally available portions of the files requested by other users can be used for interference cancellation (IC). These techniques will be reviewed in greater detail in Section \ref{section Caching Gain in RAN without Fronthaul}.


When extending coded caching to a practical RAN architecture it may not be always feasible to assume that the ENs can store all the files in the library.
%
However, this is not a limitation in practical RANs, since the ENs can fetch the missing portions of the requested contents from the cloud server via their fronthaul/backhaul connections, and then deliver them to the UEs. This network architecture is also known as a \textit{fog RAN (F-RAN)} as the ENs that are connected to the cloud processor are also endowed with storage and processing capabilities,  in contrast to cloud RANs.
The cache-aided F-RAN architecture is first considered in \cite{simeone}, where only the ENs are equipped with caches. In addition to the conventional \textit{hard-transfer} of uncached contents over the fronthaul links, the so-called \textit{soft-transfer} fronthauling \cite{simeone} is also available in an F-RAN, in which the quantized and compressed versions of the baseband signals that will be transmitted by the ENs are delivered over the fronthaul links. Cache-aided F-RANs  will be discussed in detail in Sections \ref{s:Dedicated} and \ref{section Caching gain in F-RAN with wireless fronthaul}, with dedicated and shared wireless fronthaul links, respectively.



The purpose of this paper is two-fold. First, we provide a tutorial overview of some of the existing works on coded caching in cache-aided RANs without fronthaul in Section \ref{section Caching Gain in RAN without Fronthaul}. Through intuitive examples, we demonstrate how coded caching can opportunistically enable interference cancellation, zero-forcing, and interference alignment in a wireless RAN when it is equipped at both ENs and UEs. Then we present two new caching and delivery strategies, one for a RAN with dedicated fronthaul \cite{roig2018storage} in Section \ref{s:Dedicated} and the other for a RAN with wireless fronthaul \cite{xu-isit18} in Section \ref{section Caching gain in F-RAN with wireless fronthaul}, and compare them with existing works in their corresponding sections.
In a RAN with dedicated fronthaul, we consider both centralized and decentralized cache placement at UEs, while cache placement at the ENs is centralized. We propose a new delivery scheme based on the techniques introduced in \cite{gunduz} and the soft-transfer delivery scheme of \cite{simeone}. This achievable scheme aims to minimize the delivery latency by taking into account the interplay between the EN caches, UE caches, and the fronthaul capacity. The proposed delivery scheme jointly exploits IA, ZF, IC as well as the fronthaul links. 
In a RAN with wireless fronthaul, we consider decentralized cache placement at both the ENs and UEs. In our proposed delivery scheme, the broadcast nature of the wireless fronthaul is exploited not only for fetching uncached contents, but also for fetching contents already cached at some but not all the ENs to boost EN cooperation over the access transmission. We also show that this delivery scheme is information-theoretically order-optimal. By putting the tutorial overview for cache-aided RANs without fronthaul and the  new contributions for cache-aided RANs with dedicated and wireless fronthaul together, this paper provides a comprehensive and unified treatment of content  caching and delivery in wireless RANs. Discussions for future research will also be provided.


The remainder of this paper is organized as follows. Section \ref{section system model} introduces the cache-aided RAN model, the performance metric, and how the caching and delivery of the files are carried out. Section \ref{section Caching Gain in RAN without Fronthaul} studies the cache-aided RAN without fronthaul links, and introduces some basic cache-aided interference management techniques by reviewing existing works. Section \ref{s:Dedicated} considers a cache-aided RAN with dedicated fronthaul links, and proposes a novel caching and delivery scheme with centralized caching at the ENs and centralized/decentralized caching at the UEs. Section \ref{section Caching gain in F-RAN with wireless fronthaul} studies cache-aided RAN with wireless fronthaul, and proposes a novel delivery scheme with decentralized caching at both the EN and UE sides. Section \ref{section conclusion} concludes this paper and discusses directions for future research.

\textbf{Notations:} For $K\in\mathbb{Z}^+$, $[K]$ denotes the set $\{1,2,\ldots,K\}$. For $a<b$, $a,b\in\mathbb{Z}^+$, $[a:b]$ denotes the set $\{a,a+1,\ldots,b-1,b\}$. For $x\in\mathbb{R}$, $\lfloor x\rfloor$ denotes the largest integer not greater than $x$. $(x_j)^K_{j=1}$ denotes the vector $(x_1,x_2,\cdots,x_K)^T$. We define $(x)^+\triangleq\max\{0,x\}$, and use $(\cdot)^T$ to denotes the transpose of a matrix. $A_{1\sim s}$ denotes the set $\{A_1,A_2,\ldots,A_s\}$. $\mathcal{CN}(0,1)$ denotes the complex-valued Gaussian distribution with zero mean and unit variance. $H(X)$ denotes the entropy of random variable $X$.

\section{System Model}\label{section system model}
\subsection{Network Model}\label{section network model}
\begin{figure}[!tbp]
\begin{minipage}[t]{1\linewidth}
\centering
\includegraphics[scale=0.23]{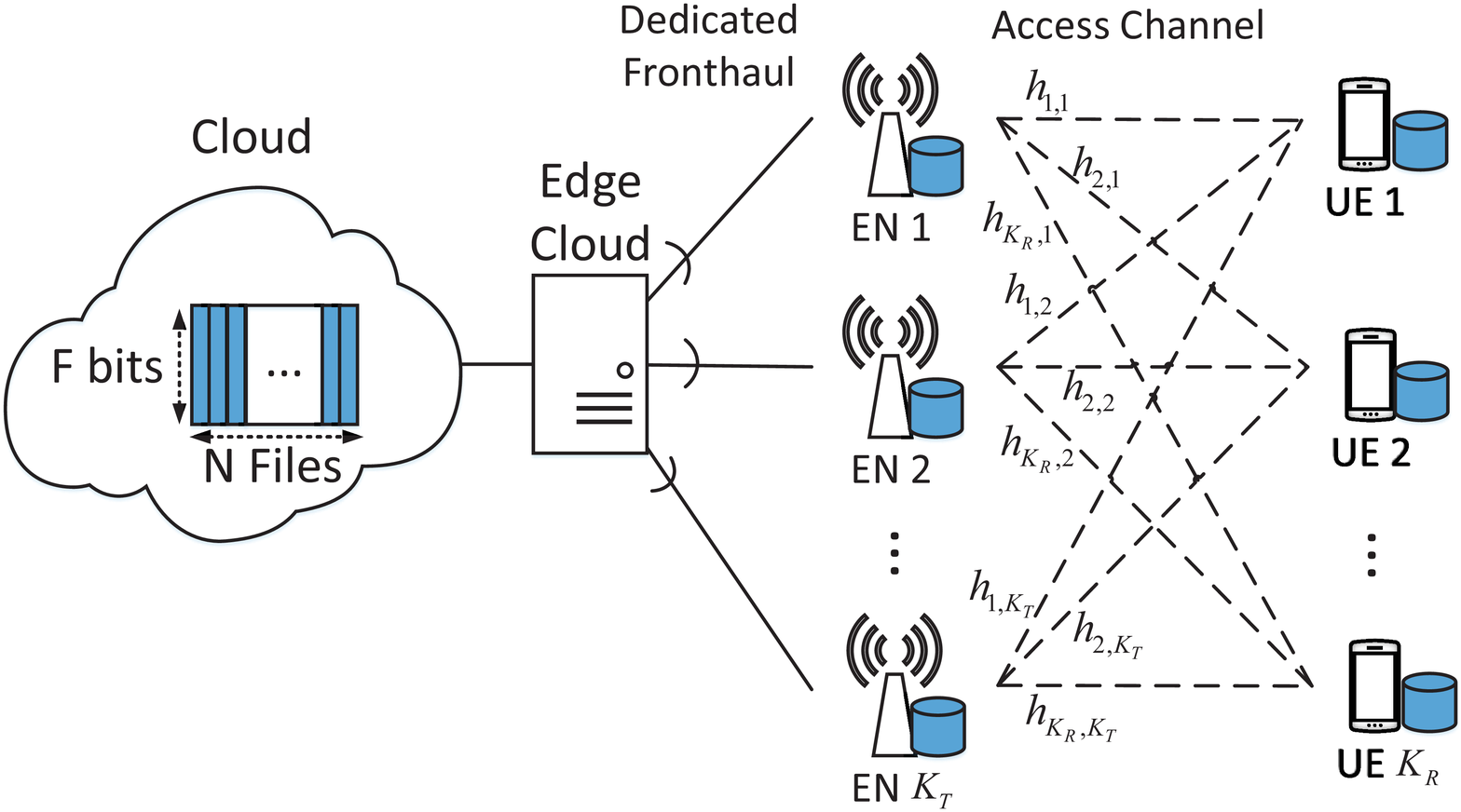}
\subcaption{}\label{Fig model dedicated}
\end{minipage}
\begin{minipage}[t]{1\linewidth}
\centering
\includegraphics[scale=0.23]{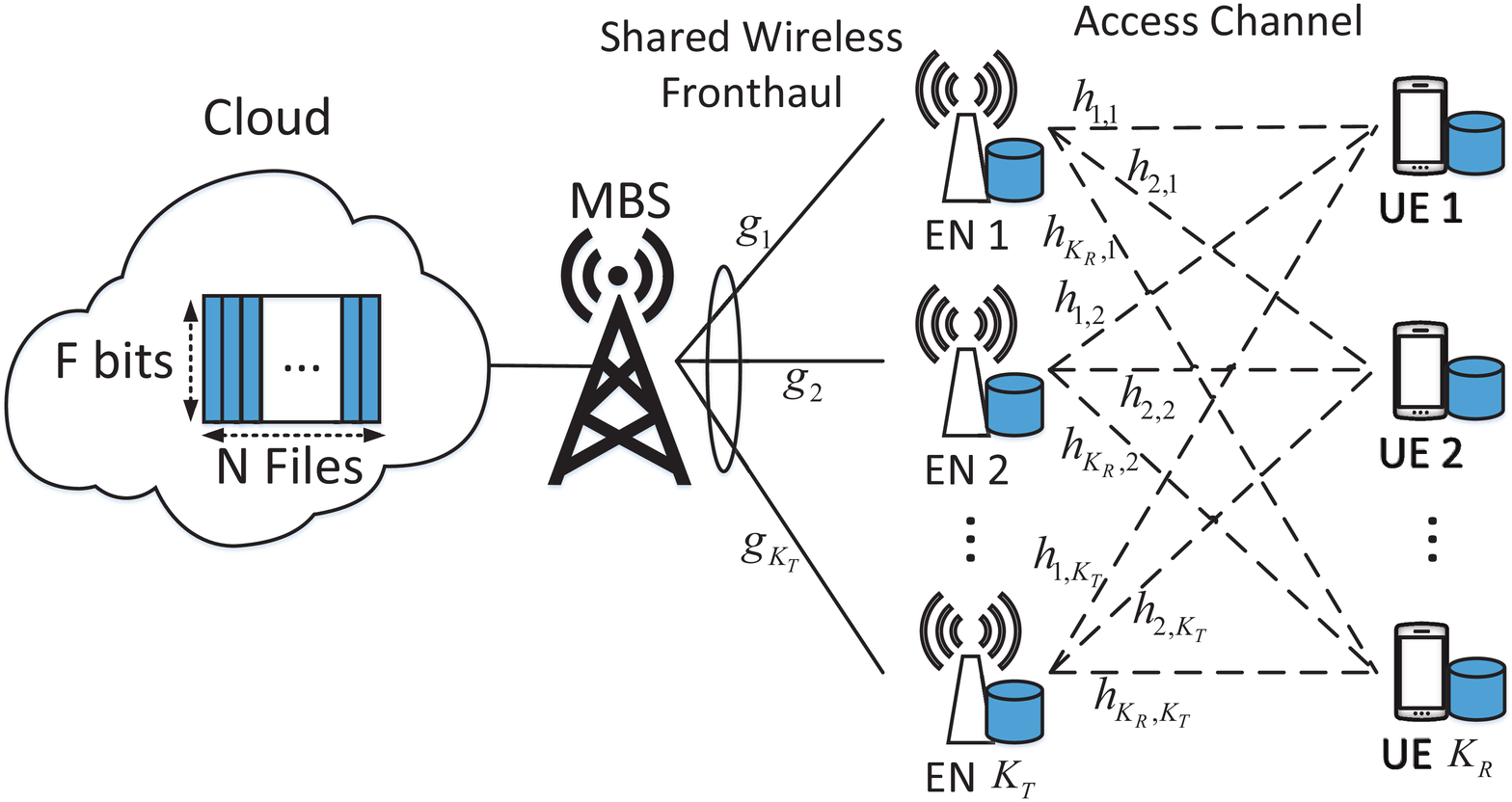}
\subcaption{}\label{Fig model wireless}
\end{minipage}
\caption{Cache-aided RAN architecture with (a) dedicated fronthaul links, (b) a shared wireless fronthaul.}\label{Fig model}
\vspace{-0.5cm}
\end{figure}

We consider a cache-aided RAN, in which $K_T\ge2$ ENs serve $K_R\ge2$ UEs over a wireless access channel, while the ENs are connected to the cloud either through dedicated fronthaul links operated by an edge cloud, as shown in Fig.~\ref{Fig model dedicated}, or through  a shared wireless fronthaul link served by a MBS, as shown in Fig.~\ref{Fig model wireless}. In the special case when the fronthaul capacity goes to zero,  the network reduces to the cache-aided interference network as studied in \cite{bothcache,niesen,mine,gunduz}, and referred to as a cache-aided RAN without fronthaul in this paper. In this paper, we will refer to the fronthaul connection from the cloud to the ENs as the \textit{fronthaul network}, and to the wireless access channel from the ENs to the UEs as the \textit{access network}.
In the fronthaul network, the capacity of the dedicated fronthaul link for each EN in bits per channel use is denoted as $C_F$ ; the channel vector of the wireless fronthaul link is denoted as $\mathbf{g}$ ,be the channel vector, where its $p$-th element for $p\in[K_T]$, denoted by $g_p\in\mathbb{C}$, is the channel coefficient from the MBS to EN $p$. In the access network, let $\mathbf{H}$ denote the channel matrix, where its $(q,p)$-th entry for $q \in [K_R]$, $p \in [K_T]$, denoted by $h_{q,p}\in\mathbb{C}$, is the channel coefficient from EN $p$ to UE $q$. For simplicity, only a single antenna is considered at all nodes. Each $g_p$ and $h_{q,p}$ are drawn from a continuous distribution and independent of each other. Throughout this paper, the channel state information $\mathbf{g}$ in the wireless fronthaul link is known by the MBS and all the ENs, but not the UEs, in the system, while the channel state information $\mathbf{H}$ is globally known within the system. We assume that a library of $N$ popular files, denoted by $\mathcal{W}=\{W_1,W_2,\ldots,W_N\}$, each of size $F$ bits, is available in the cloud. Each EN and each UE is equipped with a local cache memory that can store $\mu_TNF$ and $\mu_RNF$ bits, respectively, where $\mu_T$ and $\mu_R$ ($0\le\mu_T,\mu_R\le1$) are referred to as the \textit{normalized cache sizes} at the ENs (transmitters) and the UEs (receivers), respectively. The normalized cache size represents the fraction of the entire library that can be stored in the corresponding local cache.
%

%

The network operates in two phases, a \textit{cache placement phase} and a \textit{content delivery phase}.  In the cache placement phase, which is assumed to occur during off-peak traffic hours and over a relatively large time scale, the ENs and UEs fill their local caches. Cache placement across the network can either be done in a \textit{centralized} or a \textit{decentralized} manner. In the former, it is assumed that a central authority, e.g., the cloud server, decides what to cache in each of the caches in the network; while in the latter, each node decides its own cache contents locally and independently of the others.
In the content delivery phase, which occurs during peak traffic periods and over a shorter time scale, UE $q$, $q\in[K_R]$, requests file $W_{d_q}$, $d_q\in[N]$. We define $\mathbf{d}\triangleq(d_q)^{K_R}_{q=1}\in[N]^{K_R}$ as the \textit{demand vector}. Note that, if the ENs have collectively cached all the files in the library, then they can directly deliver users' requests over the access network without the need for fetching additional information from the cloud via the fronthaul network. Even in this case, the fronthaul network can still be utilized to deliver contents to the ENs to improve the performance over the access network, e.g., through cooperative transmission. The delivery process is a two-hop transmission, with the first hop over the fronthaul network and the second over the access network, with the aim of satisfying all the users' demands with the minimum latency possible. We consider both  \textit{full-duplex} and \textit{half-duplex} delivery schemes, where the former assumes that the ENs can transmit over the access network while receiving over the fronthaul network at the same time, while the latter assumes that the ENs either transmit or receive at any point in time, but not both simultaneously.

We next define the caching and delivery functions for this network, which are also summerized in Table \ref{table functions}.

\begin{table*}[t]
\centering
\caption{Caching, encoding, and decoding functions used across the network}
\label{table functions}
\begin{tabular}{|c|c|}
\hline
\textbf{Function} & \textbf{Notation} \\ \hline
Caching function at EN $p$ & $\phi_p: \mathcal{W}\rightarrow U_p$, for $p\in[K_T]$ \\ \hline
Caching function at UE $q$ & $\psi_q: \mathcal{W}\rightarrow V_q$, for $q\in[K_R]$ \\ \hline
Encoding function at cloud & \(\displaystyle\Lambda_F: \left\{\begin{array}{ll}\mathcal{W}, \textbf{U}, \textbf{V}, \textbf{d}, C_F, \textbf{H}\rightarrow \{\textbf{S}_{p}^{T_F}\}_{p=1}^{K_T},\text{ for dedicated fronthaul}\\ \mathcal{W}, \textbf{U}, \textbf{V}, \textbf{d},\textbf{g}, \textbf{H}\rightarrow \textbf{S}^{T_F},\text{ for wireless fronthaul}\end{array}\right.\notag\) \\ \hline
Encoding function at EN $p$ & $\Lambda^p_A: U_p,\mathbf{Q}_{p}^{T_F},\textbf{d},\textbf{H}\rightarrow \mathbf{X}_{p}^{T_A}$, for $p\in[K_T]$ \\ \hline
Decoding function at UE $q$ & $\Gamma^q_A: \mathbf{Y}_{q}^{T_A}, V_q, \textbf{d}, \textbf{H}\rightarrow \hat{W}_{d_q}$, for $q\in[K_R]$ \\ \hline
\end{tabular}
\end{table*}

%
%

%

\subsubsection{Caching functions}  Caching functions $\left\{ \phi_{p} \right\}^{K_T}_{p=1}$ and $\left\{ \psi_{q}\right\} ^{K_R}_{q=1}$ map the entire library into the cache contents at ENs and UEs, respectively. More precisely, for $p\in[K_T]$, caching function $\phi_{p}$ maps the library $\mathcal{W}$ to the cache content of EN $p$ denoted by $U_{p}$, as $U_{p}= \phi_{p} (\mathcal{W})$, where $U_{p}$ is a binary sequence of length no more than $\mu_TNF$ bits. Likewise, UE $q$, $q \in [K_R]$, employs the  caching function $\psi_{q}$ to map the library to a binary sequence of length no more than $\mu_RNF$ bits, denoted by $V_{q}$, where $V_{q}= \psi_{q} (\mathcal{W})$. We note that the caches at ENs and UEs are filled without the knowledge of the future user demands or the channel conditions during the delivery phase. We define $\textbf{U}\triangleq \left( U_p\right)^{K_T}_{p=1}$ and $\textbf{V}\triangleq \left( V_q\right)^{K_R}_{q=1}$.

In general, the caching functions $\left\lbrace \phi_{p} \right\rbrace^{K_T}_{p=1}$ and $\left\lbrace\psi_{q}\right\rbrace ^{K_R}_{q=1}$ can allow arbitrary coding within and across the files; however, many previous works, such as \cite{mine,simeone,koh2017cloud}, consider only intra-file coding, and do not allow inter-file coding, since intra-file coding, in general, can offer order-wise optimal performance. Without inter-file coding, we can rewrite the cache contents $U_p$ of EN $p$ and $V_q$ of UE $q$ as consisting of contributions from different files in the library:
\begin{align}
  U_p=\left(U_{p,1},U_{p,2},\ldots,U_{p,N}\right),\quad V_q=\left( V_{q,1},V_{q,2},\ldots,V_{q,N}\right),\notag
\end{align}
where $U_{p,n}$ and $V_{q,n}$, $n \in [N]$, denote the cache contents generated as a function of file $W_n$ at EN $p$ and UE $q$, respectively. Moving further, when neither inter-file  nor intra-file coding is applied as in \cite{bothcache}, the cache contents at EN $p$ and UE $q$ can be written, respectively, as follows:
\begin{align}
  &U_p=\{W_{n,\Phi,\Psi}: \Psi\ni p, \Psi \subseteq [K_T], \Phi \subseteq [K_R], n\in[N]\},\notag\\
  &V_q=\{W_{n,\Phi,\Psi}: \Phi\ni q, \Psi \subseteq [K_T], \Phi \subseteq [K_R], n\in[N]\},\notag
\end{align}
where $W_{n,\Phi,\Psi}$ denotes the subfile of file $W_n$ cached at UEs in set $\Phi\subseteq[K_R]$ and ENs in set $\Psi\subseteq[K_T]$ in an uncoded manner. This is commonly known as \emph{uncoded prefetching}.  This paper mainly focuses on uncoded prefetching due to its implementation simplicity and good performance.
We shall review both the centralized and decentralized cache placement with uncoded prefetching in the next subsection.

\subsubsection{Delivery functions}  The delivery scheme is defined by an  encoding function  $\Lambda_{F} $  at the cloud for transmission over the  fronthaul network, a set of encoding functions $\{\Lambda_{A}^p \} ^{K_T}_{p=1}$ at the ENs for the access network, and a set of decoding functions $\{\Gamma_{A}^q \}_{q=1}^{K_R}$ at the UEs for the access network. Unlike the caching functions, the delivery functions depend on the user demand vector  $\mathbf{d}$ and the channel matrix $\mathbf{H}$.

When a dedicated fronthaul link is used for each EN,  the cloud employs the encoding function $\Lambda_{F}$ to generate  a set of codewords  $\textbf{S}_{p}^{T_F}\triangleq \{(S_p(t))^{T_F}_{t=1}\}_{p=1}^{K_T} = \Lambda_{F} \left( \mathcal{W}, \textbf{U}, \textbf{V}, \textbf{d}, C_F, \textbf{H}\right)$. each of length $T_F$ and to be transmitted to EN $p\in[K_T]$. Here, the codeword length $T_F$ means that the transmission of each codeword $S_p^{T_F}$ takes $T_F$ channel uses. Note that $T_F$ can be zero when the fronthaul-link is deemed unnecessary. Given that the capacity of each dedicated fronthaul link is $C_{F}$ bits per channel use, not more than $T_F\cdot C_{F}$ bits can be delivered to each EN during the delivery phase. Let $Q^D_p(t)$ denote the received signal at EN $p$ at time $t$ over the dedicated fronthaul link.
\par
When the ENs share a wireless fronthaul link, the cloud employs the encoding function $\Lambda _{F}$ to generate a common codeword $\textbf{S}^{T_F}\triangleq (S(t))^{T_F}_{t=1} =\Lambda_{F}\left( \mathcal{W},\textbf{U},\textbf{V},\textbf{d},\textbf{g},\textbf{H}\right)$ of length $T_F$. The input-output relationship of the wireless fronthaul link for each symbol is modeled as:
\begin{align}
Q^W_p(t)=g_{p}S(t)+N_p(t),\ p \in [K_T], \label{eqn fronthaul model 2}
\end{align}
where $Q^W_p(t)\in \mathbb{C}$ denotes the received signal at EN $p$, $S(t)\in \mathbb{C}$ denotes the transmitted signal from the MBS subject to an average power constraint of $P_F$, i.e., $1/T_F \sum _{t=1}^{T_F}|S(t)|^2\leq P_F$, and $N_p(t)$ denotes the additive noise at EN $p$ distributed with $\mathcal{CN}(0,1)$ at time $t$.
\par
Based on the received signals from the fronthaul $\mathbf{Q}_p^{T_F} \triangleq ( Q^D_p(t) )_{t=1}^{T_F}$  or   $\mathbf{Q}_p^{T_F} \triangleq ( Q^W_p(t) )_{t=1}^{T_F}$, and the locally cached content $U_p$,  EN $p$ employs the encoding function $\Lambda_{A}^p$ to generate a codeword of length $T_A$, $\mathbf{X}_{p}^{T_A}\triangleq (X_{p}(t))^{T_{A}}_{t=1} =\Lambda_{A}^{p}\left( U_p,\mathbf{Q}_{p}^{T_F},\textbf{d}, \textbf{H}\right)$.
\par
The input-output relationship of the access network at each time slot $t$ is modeled as:
  \begin{align}
Y_q(t) = \sum_{p=1}^{K_T}h_{q,p} X_p(t) + Z_q(t),\notag
\end{align}
where $Y_q(t)\in \mathbb{C}$ denotes the received signal at UE $q$, $X_p(t)\in \mathbb{C}$ denotes the transmitted signal at EN $p$ subject to an average power constraint $P$, i.e.,  $1/T_A \sum ^{T_A}_{t=1}|X_p(t)|^2 \leq P$, and $Z_q(t)$ denotes the independent noise at UE $q$ distributed with $\mathcal{CN}(0,1)$.

Based on the received signal vector $\mathbf{Y}_{q}^{T_A}\triangleq (Y_{q}(t))^{T_{A}}_{t=1}$, the locally cached content $V_q$, the demand vector $\textbf{d}$, and the knowledge of the channel gain matrix $\mathbf{H}$, UE $q$ employs the decoding function $\Gamma_A^q$  to decode its requested file  $\hat{W}_{d_q}$. We have  $\hat{W}_{d_q} =  \Gamma_A^q(\mathbf{Y}_{q}^{T_A}, V_q,\textbf{d}, \textbf{H})$.

Note that these caching and delivery functions remain unchanged as long as the network do not change. The worst-case error probability of the system is defined as
\begin{eqnarray}
P_\epsilon=\max\limits_{\mathbf{d}\in[N]^{K_R}}\max\limits_{q\in[K_R]}\mathbb{P}(\hat{W}_{d_q}\ne W_{d_q}).\notag
\end{eqnarray}
A sequence of caching and delivery functions, consisting of $\lbrace \phi_p\rbrace ^{K_T}_{p=1}$, $\lbrace \psi_q\rbrace ^{K_R}_{q=1}$, $\Lambda _{F}$, $\left\lbrace \Lambda _{A}^{p}\right\rbrace ^{K_T}_{p=1}$, $\left\lbrace \Gamma _{A}^q\right\rbrace ^{K_R}_{q=1}$, and indexed by file size $F$, is said to be \textit{feasible} if, for almost all channel realizations, $P_\epsilon\rightarrow 0$ as $F\rightarrow\infty$. We note that the definition of the error probability imposes reliable decoding at all the UEs for all demand combinations for feasibility.

\subsection{Cache Placement}\label{section cache model}
In this subsection, we review two uncoded prefetching schemes in the placement phase, namely \emph{centralized} and \emph{decentralized} cache placement. These two schemes differ in the way the cache contents are decided across different users during the placement phase. While both schemes perform file splitting, we do not allow intra or inter-file coding.

\subsubsection{Centralized cache placement}\label{section centralized cache placement}
In centralized cache placement, the caching functions of the ENs and UEs are jointly designed by a central controller that has full knowledge of all the ENs and UEs in the network.  In practice, this means that the ENs and UEs that will participate in the delivery phase are known in advance, and their cache contents can be designed accordingly. Note that all the ENs can collectively store at most $K_T\mu_TNF$ bits from the library. In the absence of the fronthaul network, or when the fronthaul capacity goes to zero, to be able to satisfy all demand combinations, we must have $\mu_RNF+K_T\mu_TNF\ge NF$, i.e., the cache capacities of all the ENs together with the cache capacity of each single UE must be sufficient to recover all the files in the library. Equivalently, the normalized cache sizes must satisfy $\mu_R+K_T\mu_T\ge1$.  This constraint is not required in the presence of the fronthaul network.
%
%
%
%
In the following, we first introduce the \emph{symmetric} file splitting and caching scheme \cite{bothcache}, suitable for a normalized cache size pair $(\mu_R, \mu_T)$, referred to as \textit{integer points}, for which the cumulative cache capacities at both the ENs and the UEs are integers, i.e., $K_R\mu_R=i\in[0:K_R]$ and $K_T\mu_T=j\in[0:K_T]$. We then discuss the more general \emph{parametric} file splitting and caching scheme \cite{mine} suitable for arbitrary cache capacity values.

\textbf{Symmetric file splitting and caching}:  For any integer-point cache size pair $(\mu_R=\frac{i}{{K_R}}, \mu_T=\frac{j}{{K_T}})$, each file $W_n$, for $n\in[N]$, is split into $\binom{K_R}{i}\binom{K_T}{j}$ equal-size distinct subfiles $\{W_{n,\Phi,\Psi}:\Phi\subseteq[K_R],|\Phi|=i,\Psi\subseteq[K_T],|\Psi|=j\}$. Each subfile $W_{n,\Phi,\Psi}$ is then cached at the $i$ UEs in subset $\Phi$  and the $j$ ENs in subset $\Psi$.
Following this placement strategy, each EN caches $N\frac{\binom{K_R}{i}\binom{K_T-1}{j-1}F}{\binom{K_R}{i}\binom{K_T}{j}}=N\frac{j}{K_T}F=\mu_TNF$ bits, and each UE caches $N\frac{\binom{K_R-1}{i-1}\binom{K_T}{j}F}{\binom{K_R}{i}\binom{K_T}{j}}=N\frac{i}{K_R}F=\mu_RNF$ bits, which satisfy the cache capacity constraint with equality. We will illustrate symmetric file splitting and caching for $(\mu_R=\frac{1}{3},\mu_T=\frac{2}{3})$ and $(\mu_R=\frac{1}{3},\mu_T=\frac{1}{3})$ in Example \ref{example self interference cancellation gain} and Example \ref{example EN coordination gain} in Section \ref{section caching gains}, respectively, in a $3\times3$ network.

\textbf{Parametric file splitting and caching}: For any cache size pair ($\mu_T$, $\mu_R$), each file $W_n$ is partitioned into $2^{K_T+K_R}$ distinct subfiles, $\{W_{n,\Phi,\Psi}:\Phi\subseteq[K_R],\Psi\subseteq[K_T]\}$. Each subfile $W_{n,\Phi,\Psi}$ is cached at the UEs in subset $\Phi$ and the ENs in subset $\Psi$, for any $\Phi\subseteq[K_R]$ and $\Psi\subseteq[K_T]$. While each subfile can be of an arbitrary size, due to the symmetry among the nodes, the subfiles that are cached by the same number of ENs and the same number of UEs are set to have the same size. We denote the size of $W_{n,\Phi,\Psi}$ by $a_{r,t}F$ bits, where $t=|\Psi|$, $r=|\Phi|$, and $a_{r,t}\in[0,1]$ is the design parameter for file splitting. The parameters $\{a_{r,t}\}$ should satisfy the following constraints:
\begin{empheq}[left=\empheqlbrace]{align}
&\sum_{r=0}^{K_R} \sum_{t=0}^{K_T}\binom{K_R}{r}\binom{K_T}{t}a_{r,t}=1,\label{eqn:total cache}\\
&\sum_{r=1}^{K_R} \sum_{t=0}^{K_T}\binom{K_R-1}{r-1}\binom{K_T}{t}a_{r,t}\le\mu_R,\label{eqn:receiver cache}\\
&\sum_{r=0}^{K_R} \sum_{t=1}^{K_T}\binom{K_R}{r}\binom{K_T-1}{t-1}a_{r,t}\le\mu_T.\label{eqn:transmitter cache}
\end{empheq}
Here, constraint \eqref{eqn:total cache} guarantees that all $F$ bits of each file are considered, while constraints \eqref{eqn:receiver cache} and \eqref{eqn:transmitter cache} guarantee  that  the receiver and transmitter cache capacities are not violated, respectively. Note that  these constraints differ slightly from the original constraints in \cite{mine} in that the parameters $a_{r,0}$, for $0\le r<K_R$ are present in \eqref{eqn:total cache} and \eqref{eqn:receiver cache} but not in \cite[Eq. (13) and (14)]{mine}. This is because \cite{mine} assumed that every bit of a file must be either cached in at least one EN or cached in all UEs if not in any EN due to the absence of the fronthaul network. We will illustrate parametric file splitting and caching in Example \ref{example ZF gain parametric} in Section \ref{section caching gains} for  $(\mu_R=\frac{1}{3},\mu_T=\frac{2}{3})$ in a $3\times3$ network with file splitting parameters $a_{3,0}=\frac{1}{3},a_{0,3}=\frac{2}{3}$ and others being 0.
We note that parametric file splitting and caching scheme is more general than the symmetric one as it considers all possible cache placement combinations. But, bear in mind that not all cache combinations are actually needed in a given system as some of the file splitting parameters $\{a_{r,t}\}$ can be zero after optimization. The symmetric scheme is a special case of the parametric one with $a_{i,j} = 1/ \binom{K_R}{i}\binom{K_T}{j}$ and $a_{r,t}=0$, $\forall r\ne i, \forall t\ne j$ at integer-point cache size pair $(\mu_R=\frac{i}{{K_R}}, \mu_T=\frac{j}{{K_T}})$, for $i\in [0:K_R]$ and $j \in [0:K_T]$.

\subsubsection{Decentralized cache placement}\label{section decentralized cache placement}
Throughout this paper, we limit the decentralized caching strategy to the random decentralized cache placement as originally proposed  in \cite{fundamentallimits} only, though other decentralized strategies are also possible.  In specific, each cache node in the network independently caches $\mu F$ bits chosen uniformly at random from each of the $N$ files in the library,  with $\mu$ being its normalized cache size.  This scheme is particularly suitable for a large and random network where it is demanding for a central controller to coordinate the contents of too many caches, or when the identity and the number of users that will take part in the delivery phase are unknown at the placement phase, e.g., mobile users randomly connecting to access points.
%
%
 %
As a result of the randomness, the size of each subfile cached simultaneously in a given set of nodes is a random variable. However, for a sufficiently large file size $F$, the size of each subfile cached exclusively by an arbitrary set of $r$ UEs, $r\in[0:K_R]$, and an arbitrary set of $t$ ENs, $t\in[0:K_T]$, converges to $\mu_R^r (1-\mu_R)^{K_R-r}\mu_T^t(1-\mu_T)^{K_T-t}F+o(F)$ bits with high probability \cite{Wan2018}. Thus, for ease of analysis, we ignore the $o(F)$ term and define
\begin{align}
f_{r,t}\triangleq\mu_R^r (1-\mu_R)^{K_R-r} \mu_T^t(1-\mu_T)^{K_T-t}.\label{eqn decentralized subfile size}
\end{align}
as the fractional size of each subfile cached at an arbitrary set of $r$ UEs and an arbitrary set of $t$ ENs when $F\rightarrow\infty$. If the random caching scheme is employed by the UEs only, the fractional size of each subfile cached simultaneously at an arbitrary set of $r$ UEs, when $F\rightarrow\infty$, is given by:
\begin{align}
f_{r}\triangleq\mu_R^r(1-\mu_R)^{K_R-r}\label{eqn decentralized receiver subfile size}
\end{align}

Note that if we adopt decentralized cache placement at both the ENs and UEs, a non-zero fronthaul connection is always required for the existence of feasible codes even when the normalized cache sizes satisfy $\mu_R+K_T\mu_T\ge1$. This is because, due to the random nature of cache placement, some bits of the requested file will not be cached by any of the ENs with a non-zero probability.


\subsection{Performance Metric}

We adopt a latency-oriented performance metric, called the \emph{normalized delivery time }(NDT), introduced in \cite{zhang2017fundamental}, \cite{lowerbound}, \cite{simeone}, and widely used in the literature \cite{mine,gunduz,cao,girgis2017decentralized,goselingdelivery,koh2017cloud,Cran}. It is defined as the worst-case latency required to serve any possible user demand vector $\mathbf{d}$, normalized by the required time to transmit a single file in a point-to-point baseline channel, in the high signal-to-noise ratio (SNR) regime.
As noted in \cite{mine, gunduz}, NDT captures not only the improvement in channel degrees-of-freedom (DoF) thanks to cache-enabled EN cooperation (as studied in \cite{upperbound}), but also the reduction in the data load that needs to be delivered to the UEs thanks to their local caches.
While this work focuses on the worst-case NDT analysis, we would like to mention that there are some works on non-worst-case NDT analysis in the literature. In particular, \cite{azimi_online} studies the average NDT over a long time horizon for online coded caching, \cite{girgis_converse} presents a lower bound of the expected NDT, and \cite{long_glo17} studies the NDT at any given (not necessarily distinct) user demand.

Based on the system model outlined in Section \ref{section network model}, the total content delivery latency, denoted as $T$, is given by $T=T_F+T_A$ for half-duplex ENs, and by $T=\max\{T_F,T_A\}$ for full-duplex ENs. In the following, we first define the delivery time per bit, and then define NDT formally.
\begin{definition}\label{definition dtpb}
\textit{Delivery time per bit} $\Delta (\mu_T,\mu_R,C_F {\rm{~or~}} P_F,P)$ is said to be \textit{achievable} for the cache-aided RAN with dedicated fronthaul links at capacity $C_F$ or with a wireless fronthaul link at power $P_F$, if there exist a sequence of feasible caching and delivery codes so that
\begin{align*}
\Delta (\mu_T,\mu_R,C_F {\rm{~or~}} P_F,P)  = \lim _{F\to \infty} \frac{T\left(T_F,T_A\right)}{F}.
\end{align*}
\end{definition}

The delivery time per bit measures the average number of channel uses required to transmit a single bit to all the UEs in the network.

Due to the difficulty of characterizing the delivery time per bit in a multi-user network at a given finite SNR value, we will resort to the high SNR analysis, which will allow us to gain insights into the potential benefits of caching and coded delivery in a two-hop RAN architecture. Accordingly, we let the fronthaul link capacity scale as $C_F=r_D\log \left(P\right)$ in the case of dedicated fronthaul links, or let the power constraint of the MBS grow as $P_F=(P)^{r_W}$ in the case of a shared wireless fronthaul link. We note that $r_D$ and $r_W$ can be viewed as the \emph{multiplexing gain} of the fronthaul link with respect to the capacity of a point-to-point EN-UE wireless channel in the high SNR regime.

\begin{definition}\label{definition ndt}
 The \textit{normalized delivery time} (NDT) of a cache-aided RAN with an achievable delivery time per bit $\Delta (\mu_T,\mu_R,C_F,P)$ with dedicated fronthaul links at capacity $C_F=r_D\log P$,  or  $\Delta (\mu_T,\mu_R,P_F,P)$ with a wireless fronthaul link at power $P_F=(P)^{r_W}$  is defined as
\begin{align}
&\tau(\mu_T,\mu_R, r)\triangleq\left\{
\begin{array}{ll}
\lim\limits_{P\to\infty}\frac{\Delta (\mu_T,\mu_R,r_D\log P,P)}{1/\log P},\text{ for dedicated fronthaul }r\!=\!r_D,\\
\lim\limits_{P\to\infty}\frac{\Delta (\mu_T,\mu_R,(P)^{r_W},P)}{1/\log P},\text{ for wireless fronthaul } r\!=\!r_W.\\
\end{array}
\right.\notag
\end{align}
Moreover, the minimum NDT $\tau^*(\mu_R,\mu_T, r)$ is defined as the infimum of $\tau(\mu_T,\mu_R, r)$ over all achievable NDT values. For simplicity we will shortly use $\tau$ and $\tau^*$ in the rest of the paper, unless we want to highlight its dependence on the parameters $\mu_T,\mu_R$, and $r$.
%
\end{definition}

%
%

Similarly to Definition \ref{definition dtpb} and Definition \ref{definition ndt}, for a sequence of feasible caching and delivery codes, we can separately define the \textit{fronthaul NDT} and the \textit{access NDT} as
\begin{align}
&\tau_F(\mu_R,\mu_T,r)\triangleq\lim_{P\to\infty}\lim_{F\to\infty}\sup\frac{T_F}{F/\log P},\notag\\
&\tau_A(\mu_R,\mu_T,r)\triangleq\lim_{P\to\infty}\lim_{F\to\infty}\sup\frac{T_A}{F/\log P},\notag
\end{align}
respectively. For half-duplex ENs, we have $\tau=\tau_F+\tau_A$; while with full-duplex EN transmission, we have $\tau=\max\{\tau_F,\tau_A\}$.


\begin{remark}\label{remark_dof}
At a given feasible sequence of caching and delivery codes, the NDTs of the fronthaul and access networks can be computed separately by counting the actual amount of information bits delivered to a particular node and the transmission rate to that node. More specifically, let $R_F\cdot F$ denote the number of bits delivered to each EN over the fronthaul link with multiplexing gain $r$, then the fronthaul NDT can be computed as $\tau_F = R_F/r$, where $r=r_D$ for dedicated fronthaul links, and $r=r_W$ for wireless fronthaul. Likewise, let $R_A \cdot F $ denote the number of bits delivered to each UE over the access link at a transmission rate of $d\cdot \log P + o(\log P)$, then the access NDT can be computed as $\tau_A = R_A/d$, where $d$ represents the standard channel DoF for that UE in the high SNR regime \cite{upperbound}.
\end{remark}

\section{Cache-Aided RAN without Fronthaul}\label{section Caching Gain in RAN without Fronthaul}

In this section, we provide a tutorial overview of some of the caching and delivery schemes in a cache-aided RAN without fronthaul connections\footnote{The absence of fronthaul connections means that the ENs cannot fetch the requested file bits from the cloud during the content
delivery phase. However, the ENs are still allowed to cooperate for data transmission. Hence, certain connections between the
ENs still exist to convey necessary signalling overhead.} The simplicity of this model allows us to introduce various cache-aided interference management techniques and understand the primary benefits of caching as well as the overall performance bounds, which will later be instrumental in studying more involved and practical RAN architectures with fronthaul connections. Throughout this section, we focus on the cache size region $\mu_R+K_T\mu_T\ge1$ since there is no fronthaul connection and therefore the accumulated cache capacity at all ENs together with the cache capacity at each single UE should be large enough to collectively store the entire file library.

\subsection{Caching Gains}\label{section caching gains}

In this subsection we present the various gains thanks to the joint UE and EN caches through illustrative examples. The gains from the UE caches, besides the obvious local caching gain, are obtained by treating the cache content at each UE as side information for coded multicasting or by cancelling known interference. Both of these gains will be referred to as IC since the number of interfering signals at the UEs is effectively reduced thanks to the proactively cached contents. The gains from the EN caches are obtained through the elimination or reduction of the received signal space of the interference via collaborative EN transmission, such as ZF and IA.

In the following, we present three examples in a $3\times 3$ RAN to elaborate in detail how these different gains (i.e., IC, ZF, and IA) are obtained jointly or individually with proper cache placement and delivery schemes. To simplify the presentation, we only consider integer-point cache sizes (i.e., $K_T\mu _T$, $K_R \mu _R$ $\in  \mathbb{Z}$) in these examples. Note that, in general, the worst-case demand vector corresponds to each user requesting a different file from the library. We assume, without loss of generality, that UE $q$,  $q\in[3]$, requests file $W_q$ in the delivery phase.
When some UEs request the same file, the delivery schemes proposed for distinct requests can still be applied by treating the requests as different files, which, however, may cause higher transmission latency than considering the common requests explicitly.

\begin{example}[IC gain and ZF gain \cite{bothcache,gunduz,mine}]\label{example self interference cancellation gain}
Consider ($\mu_R=\frac{1}{3},\mu_T=\frac{2}{3}$). In the cache placement phase, by the symmetric file splitting and caching scheme, each file $W_n$, $n\in[N]$, is split into $9$ equal-size subfiles
\begin{align}
  &\left\{W_{n,\{1\},\{1,2\}},W_{n,\{1\},\{1,3\}},W_{n,\{1\},\{2,3\}},W_{n,\{2\},\{1,2\}},W_{n,\{2\},\{1,3\}},\right.\notag\\
  &\left. \ W_{n,\{2\},\{2,3\}},W_{n,\{3\},\{1,2\}},W_{n,\{3\},\{1,3\}},W_{n,\{3\},\{2,3\}}\right\},\label{eqn example 1}
\end{align}
where subfile $W_{n,\{q\},\Psi}$ is cached at UE $q$ and the ENs in set $\Psi$. Each UE has cached $3$ subfiles of its desired file, and needs the remaining $6$ subfiles. Therefore, there are a total of 18 subfiles to be transmitted over the access link. Each UE desires 6 out of 18 subfiles, while the remaining 12 subfiles act as interference. Note that some of these undesired subfiles are also cached by each UE; and thus can be utilized as side information for IC. For example, the undesired subfiles for UE 1 are
\begin{align}
  &\left\{W_{2,\{1\},\{1,2\}},W_{2,\{1\},\{1,3\}},W_{2,\{1\},\{2,3\}},W_{3,\{1\},\{1,2\}},W_{3,\{1\},\{1,3\}},W_{3,\{1\},\{2,3\}},\right.\notag\\
  &\ \left.W_{2,\{3\},\{1,2\}},W_{2,\{3\},\{1,3\}},W_{2,\{3\},\{2,3\}},W_{3,\{2\},\{1,2\}},W_{3,\{2\},\{1,3\}},W_{3,\{2\},\{2,3\}}\right\}.\notag
\end{align}
Here, the first 6 subfiles are already cached at UE 1, and thus can be eliminated by IC, and only the remaining 6 subfiles act as interference at UE 1.

Next, we explain how cooperative beamforming can be employed by the ENs to cancel the remaining interference at each UE. Consider, for example, subfile $W_{1,\{2\},\{1,2\}}$, which is intended for UE 1, cached at UE 2, and causes interference to UE 3. Let EN 1 and EN 2 transmit this subfile with beamforming factors $v_1=h_{3,2}$ and $v_2=-h_{3,1}$, respectively. The received signal gain for $W_{1,\{2\},\{1,2\}}$ at UE 3 thus becomes $h_{31}v_1+h_{32}v_2=0$. That is, the interference caused by $W_{1,\{2\},\{1,2\}}$ is zero-forced at UE 3. This ZF method can be applied similarly to all the subfiles by designing the corresponding beamforming factors so as to cause zero interference to their unintended UEs. As a result, each UE only receives signals for its $6$ desired subfiles without any interference, which can be decoded via a six-symbol extension, achieving a per-user DoF of $1$. Following Remark \ref{remark_dof}, an NDT of $\tau(\frac13, \frac23, 0)=\frac{6\times 1/9}{1} = \frac{2}{3}$ is thus achieved, where the numerator  accounts for the normalized total size of the subfiles intended for each user, and the denominator is the achievable DoF per user.
\end{example}

\begin{example}[ZF gain \cite{mine}]\label{example ZF gain parametric}
Consider ($\mu_R=\frac{1}{3},\mu_T=\frac{2}{3}$) again as in Example \ref{example self interference cancellation gain}.  Instead of splitting each file into $9$ equal-size subfiles, we now split file $W_n$, $n\in[N]$, into $2$ unequal-size subfiles as:
%
\begin{align}
  \left\{W_{n,\{1,2,3\},\emptyset},W_{n,\emptyset,\{1,2,3\}}\right\},\notag
\end{align}
where $W_{n,\{1,2,3\},\emptyset}$ contains $\frac{1}{3}F$ bits and is cached at all three UEs but none of the ENs, while  $W_{n,\emptyset,\{1,2,3\}}$ contains $\frac{2}{3}F$ bits and is cached at all three ENs but none of the UEs. This cache placement scheme corresponds to the parametric file splitting and caching scheme with file splitting parameters $a_{3,0}=\frac{1}{3},a_{0,3}=\frac{2}{3}$ and others being 0. Upon user requests, each UE $q$ only needs the subfile $W_{q,\emptyset,\{1,2,3\}}$ since it has cached the other. Therefore, the system only has 3 subfiles to deliver, one for each UE. The fact that the subfiles that need to be delivered are cached at all the three ENs turns the channel into a multi-input single-output (MISO) broadcast channel with each EN acting as a virtual antenna. By designing the ZF beamforming vectors at all three ENs, each subfile can be successfully decoded at its desired UE without interference. Thus, an NDT of $\tau(\frac13, \frac23, 0) = \frac{2}{3}$ can be achieved. Compared to symmetric file splitting adopted in Example \ref{example self interference cancellation gain}, the asymmetric file splitting adopted in this example enables full EN cooperation and does not require IC at the UEs, yet achieving the same NDT performance.
\end{example}

\begin{example}[IC gain and IA gain \cite{niesen,myisit}]\label{example EN coordination gain}
Consider ($\mu_R=\frac{1}{3},\mu_T=\frac{1}{3}$).
By using symmetric file splitting and caching, file $W_n$, $n\in[N]$, is split into 9 equal-size subfiles:
\begin{align}
  \left\{W_{n,\{1\},\{1\}},W_{n,\{1\},\{2\}},W_{n,\{1\},\{3\}},W_{n,\{2\},\{1\}},W_{n,\{2\},\{2\}},W_{n,\{2\},\{3\}},W_{n,\{3\},\{1\}},W_{n,\{3\},\{2\}},W_{n,\{3\},\{3\}}\right\},\label{eqn example 4}
\end{align}
and each subfile $W_{n,\{q\},\{p\}}$ is cached at UE $q$ and EN $p$. Each UE caches 3 subfiles of its requested file and needs the remaining 6 subfiles, resulting in a total of 18 subfiles to be transmitted.  Each UE desires 6 out of the 18 subfiles and sees the other 12 subfiles as interference.  Given that each of these $18$ subfiles is desired by one UE and cached at another UE, we can perform pair-wise XOR combining and shrink the set of 18 subfiles to a set of 9 coded messages:
\begin{align}
  \left\{W_{\{1,2\},\{1\}}^\oplus,W_{\{1,3\},\{1\}}^\oplus,W_{\{2,3\},\{1\}}^\oplus,W_{\{1,2\},\{2\}}^\oplus,W_{\{1,3\},\{2\}}^\oplus,W_{\{2,3\},\{2\}}^\oplus,W_{\{1,2\},\{3\}}^\oplus,W_{\{1,3\},\{3\}}^\oplus,W_{\{2,3\},\{3\}}^\oplus\right\},\label{eqn example 4 2}
\end{align}
where $W_{\{q_1,q_2\},\{p\}}^\oplus\triangleq W_{q_1,\{q_2\},\{p\}}\oplus W_{q_2,\{q_1\},\{p\}}$ is generated at EN $p$  and desired by UEs  $q_1,q_2$.  Now, each UE desires 6 out of the 9 coded messages, while the remaining 3 coded messages act as interference. The number of interfering signals at each UE is thus reduced via exploiting coded multicasting.

Next, we show how to coordinate the beamforming design at the ENs through the IA technique to align the 3 undesired messages along the same direction at each UE. Consider UE 1 as an example. The undesired messages at UE 1, $W_{\{2,3\},\{1\}}^\oplus$, $W_{\{2,3\},\{2\}}^\oplus$, and $W_{\{2,3\},\{3\}}^\oplus$, are precoded with beam-forming vectors $\mathbf{v}_{\{2,3 \}, 1}, \mathbf{v}_{\{2,3 \}, 2}$ and $\mathbf{v}_{\{2,3 \}, 3}$, respectively. To align these messages, we choose the beamforming vectors such that $\mathbf{H}_{1,1}\mathbf{v}_{\{2,3\},1}=\mathbf{H}_{1,2}\mathbf{v}_{\{2,3 \}, 2}=\mathbf{H}_{1,3}\mathbf{v}_{\{2,3\}, 3}$, where $\mathbf{H}_{1,p}$ is the channel matrix between EN $p$ and UE $1$ after certain symbol extension is applied. Similar beamforming vectors are applied to align the interference at UE 2 and UE 3 as well.
As a result, each UE can decode 6 desired messages and suppress 3 undesired messages that are aligned in the same subspace via 7-symbol extension, yielding a per-user DoF of $\frac{6}{7}$ .  Finally, an NDT of $\tau(\frac13,\frac13, 0) =\frac{2/3}{6/7} = \frac{7}{9}$  can be achieved.
\end{example}

Through the above examples, we have demonstrated that caching accelerates content delivery over the access network in a RAN architecture by opportunistically (depending on cache capacities) changing the information flow, and by enabling various interference management techniques.  These include IC or coded multicasting for subfiles cached at UEs, ZF for subfiles cached at multiple ENs, and IA for subfiles cached at only one EN.

\subsection{Performance Bounds}

Following the above illustrative examples,  several achievable upper bounds on the optimal NDT of a general  $K_T\times K_R$ cache-aided RAN without fronthaul are obtained in  \cite{bothcache, niesen, mine,gunduz}.  The works  \cite{bothcache, niesen, mine} also provide theoretical lower bounds on the optimal NDT, but under different constraints and assumptions.
In the subsection, we present these bounds and provide some comparison and discussion.

Using symmetric file splitting and cache placement and exploiting IC and ZF (one-shot linear delivery) as in Example \ref{example self interference cancellation gain}, the authors in \cite{bothcache} show that the following NDT is achievable:
\begin{align}
  \tau_{\rm{NMA}}=\frac{K_R(1-\mu_R)}{\min\{K_R,K_R\mu_R+K_T\mu_T\}}\label{eqn NMA scheme}
\end{align}
at an arbitrary integer-point cache size pair $(\mu_R=i/K_R,\mu_T=j/K_T)$, for $i\in [0:K_R],j\in [K_T]$. In \eqref{eqn NMA scheme}, $(1-\mu_R)$ stems from the local caching gain at the UEs, while $\min\{K_R,K_R\mu_R+K_T\mu_T\}$ is the achievable sum DoF in the delivery phase, where the $K_R\mu_R$ term is due to IC  at the UE side, and the $K_T\mu_T$ term is due to the ZF gain at the EN side.
Using the same symmetric file splitting and caching strategy but exploiting coded muticasting and IA as in Example \ref{example EN coordination gain} , the authors in \cite{niesen} show that the following minimum NDT is achievable:
\begin{align}
  \tau_{\rm{HND}}=\frac{K_T-1+\frac{K_R}{K_R\mu_R+1}}{K_T}(1-\mu_R)\label{eqn HND scheme}
\end{align}
at an arbitrary integer-point cache size pair $(\mu_R=i/K_R,\mu_T=j/K_T)$, for $i\in[0:K_R],j\in[K_T]$. The fact that  $\tau_{HND}$  is independent of $\mu_T$  is because each EN in \cite{niesen} only caches $F/K_T$ bits of each file from the library without overlap, regardless of its actual normalized cache size $\mu_T > 1/{K_T}$. In \eqref{eqn HND scheme}, $(1-\mu_R)$ results from the local caching gain at UEs, similar to \eqref{eqn NMA scheme}, and $({K_T-1+\frac{K_R}{K_R\mu_R+1}})/{K_T}$ results from the combined coded multicasting and IA gain. Using the symmetric file splitting and caching, again, the authors  in \cite{gunduz} proposed another scheme exploiting IC, ZF and IA jointly.
However, the expression provided in \cite{gunduz} is not valid for all network configurations, as aligned messages are not always guaranteed to be decodable due to the limit degrees of freedom available. The NDT presented in \cite{gunduz} holds for the 3x3 RAN and is included in the numerical comparison provided below. For all the schemes proposed in \cite{bothcache, niesen,gunduz}, the minimum NDT at non-integer cache size points can be obtained through the memory sharing techniques \cite{fundamentallimits}.

Using  parametric file splitting and caching, and optimizing the file splitting parameters as in Example \ref{example ZF gain parametric} ,  the authors in \cite{mine} show that the minimum NDT obtained by solving the following linear program (LP)  is achievable:
\begin{align}
\tau_{\rm{XTL}}\triangleq  \min\sum_{r=0}^{K_R-1}&\sum_{t=1}^{K_T}\frac{\binom{K_R-1}{r}\binom{K_T}{t}}{d_{r,t}}a_{r,t},\label{eqn tmin}\\
\textrm{s.t.} \qquad& \eqref{eqn:total cache}, \eqref{eqn:receiver cache}, \eqref{eqn:transmitter cache}\\
&0\le a_{r,t} \le 1,\forall (r,t)\in \mathcal{A}
\end{align}
for any cache size pair $(\mu_R, \mu_T)$. Here, $\mathcal{A}\triangleq\{(r,t): r+K_Rt\ge K_R,0\le r\le N_R, 0\le t\le K_T,r,t\in\mathbb{Z}\}$ is the set of all possible integer pairs $(r,t)$, and $d_{r,t}$ is the achievable per-user DoF for the  $\binom{K_T}{t}\times\binom{K_R}{r+1}$ cooperative X-multicast channel\footnote{In a $\binom{K_T}{t}\times\binom{K_R}{r+1}$ cooperative X-multicast channel, each set of $r+1$ UEs forms a UE multicast group,  each set of $t$  ENs forms an EN cooperation group, and each EN cooperation group has an independent message for each UE multicast group \cite[Definition 2]{mine}. }, given by \cite[Lemma 1]{mine}
\begin{align}
d_{r,t}
=\left\{
\begin{array}{ll}
1, &r+t\ge K_R\\
\frac{\binom{K_R-1}{r}\binom{K_T}{t}t}{\binom{K_R-1}{r}\binom{K_T}{t}t+1}, &r+t= K_R-1\\
\max\left\{d_{1},\frac{r+t}{K_R}\right\}, &r+t\le K_R-2
\end{array}
, \right.\label{eqn tau ip}
\end{align}
where
\begin{align}
d_1\triangleq \max\limits_{1\le t'\le t}\!\left\{\!\frac{\binom{K_R-1}{r}\binom{K_T}{t'}\binom{K_R-r-1}{t'-1}t'}
{\binom{K_R-1}{r}\binom{K_T}{t'}\binom{K_R-r-1}{t'-1}t'\!+\!\binom{K_R-1}{r+1}\binom{K_R-r-2}{t'-1}\binom{K_T}{t'-1}}\!\right\}\!.
\end{align}
Note that unlike \cite{bothcache, niesen,gunduz}, the scheme in \cite{mine} intrinsically includes memory sharing in its formulation through the parametric file splitting. The term multiplied by each file splitting parameter $a_{r,t}$ in the objective function represents the joint IC, ZF, and IA gain.  In the  $3\times 3$ network, for example, the achievable NDT by solving the above LP is a piece-wise linearly decreasing function of the cache size pair:
\begin{align}
\tau_{\rm{XTL}}=
\left\{
\begin{array}{ll}
1-\mu_R, & (\mu_R,\mu_T)\in \mathcal{R}^1\\
\frac{4}{3}-\frac{4}{3}\mu_R-\frac{1}{3}\mu_T, & (\mu_R,\mu_T)\in \mathcal{R}^2\\
\frac{3}{2}-\frac{5}{3}\mu_R-\frac{1}{2}\mu_T, & (\mu_R,\mu_T)\in \mathcal{R}^3\\
\frac{13}{6}-\frac{8}{3}\mu_R-\frac{3}{2}\mu_T, & (\mu_R,\mu_T)\in \mathcal{R}^4\\
\frac{8}{3}-\frac{8}{3}\mu_R-3\mu_T, & (\mu_R,\mu_T)\in \mathcal{R}^5
\end{array}
\right.\label{eqn mine 3x3}
\end{align}
where $\{\mathcal{R}^i\}^5_{i=1}$ are given as
\begin{align}
\left\{
\begin{array}{ll}
\mathcal{R}^1=\{(\mu_R,\mu_T): \mu_R+\mu_T\ge1, \mu_R\le1, \mu_T\le1\}\\
\mathcal{R}^2=\{(\mu_R,\mu_T): \mu_R+\mu_T<1, 2\mu_R+\mu_T\ge1,\mu_R+2\mu_T>1\}\\
\mathcal{R}^3=\{(\mu_R,\mu_T): 3\mu_R+3\mu_T\ge2, 2\mu_R+\mu_T<1,\mu_R\ge0\}\\
\mathcal{R}^4=\{(\mu_R,\mu_T): 3\mu_R+3\mu_T<2, \mu_R\ge0, 3\mu_T>1\}\\
\mathcal{R}^5=\{(\mu_R,\mu_T): 3\mu_T\le1, \mu_R+2\mu_T\le1,\mu_R+3\mu_T\ge1\}
\end{array}.\notag
\right.
\end{align}

Next we present a lower bound on the optimal NDT. It is shown in \cite{mine} that when neither inter-file nor intra-file coding is allowed in the cache placement (i.e., uncoded prefetching), the minimum NDT is lower bounded by $\tau_L$  defined as follows \cite[eq.(11)]{mine}:
\begin{align}
\tau_{L} \triangleq\max\limits_{\substack{l \in [\min\{K_T,K_R\}]\\s_1\in [0:l]\\s_2\in [0:K_R-l]}}\frac{1}{l}
\bigg\{&(s_1+s_2)-(K_T-l)s_2\mu_T-\left(\frac{2s_2+s_1+1}{2}\cdot s_1+s_2^2\right)\mu_R\notag\\
&+\left(\frac{2s_2+s_1}{2}(s_1-1)+s_2^2\right)(1-K_T\mu_T)^+\bigg\}.\label{eqn taul2}
\end{align}
Different lower bounds are obtained in \cite{bothcache} and \cite{niesen}. But the one in \cite{bothcache} is restricted to one-shot linear delivery scheme, and thus cannot bound the performance of symbol-extension based delivery schemes, such as IA. The bound in \cite{niesen} allows arbitrary intra- and inter-file coding, thus it is not as tight as \eqref{eqn taul2} to bound the performance of uncoded prefetching that is widely adopted in the literature.

\begin{figure}[!tbp]
\begin{minipage}[t]{1\linewidth}
\centering
\includegraphics[scale=0.4]{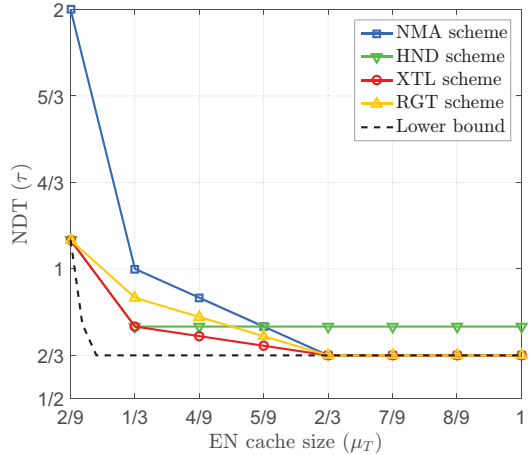}
\subcaption{}\label{Fig Comparison_in_RAN 1}
\end{minipage}
\begin{minipage}[t]{1\linewidth}
\centering
\includegraphics[scale=0.4]{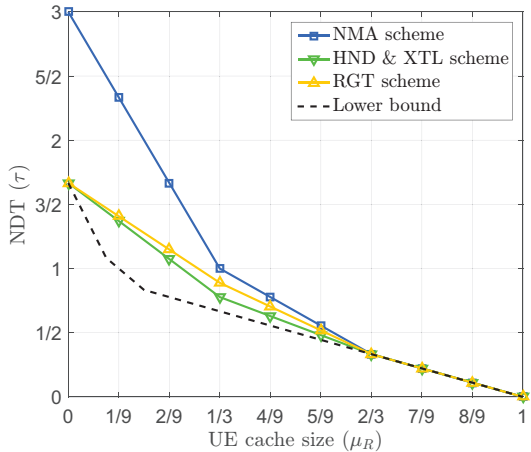}
\subcaption{}\label{Fig Comparison_in_RAN 2}
\end{minipage}
\caption{Achievable NDT in a $3\times 3$ RAN a) as a function of $\mu_T$ when $\mu_R=1/3$, and b) as a function of $\mu_R$ when $\mu_T=1/3$.}\label{Fig Comparison_in_RAN}
\vspace{-0.5cm}
\end{figure}

Finally, we compare the performance of the aforementioned achievable schemes in \cite{bothcache, niesen, mine,gunduz} along with the lower bound in \eqref{eqn taul2}.
Fig.~\ref{Fig Comparison_in_RAN} illustrates the NDTs achieved by different schemes in a $3\times3$ RAN. Note that the performance of the scheme in \cite{bothcache}, referred to as NMA, the scheme in \cite{niesen}, referred to as HND, and the scheme in \cite{gunduz}, referred to as RGT, at non-integer points is obtained via memory sharing between integer points. The performance of the scheme in \cite{mine}, referred to as XTL, on the other hand, is obtained directly from \eqref{eqn mine 3x3} at any cache size pair. It can be seen from Fig.~\ref{Fig Comparison_in_RAN 1} that the XTL scheme and the RGT scheme are optimal when $(\mu_R=1/3,\mu_T = 2/9)$ and $(\mu_R=1/3,\mu_T\ge 2/3)$. This is because both XTL and RGT schemes exploit jointly the IC, ZF, and IA gains, while the HND scheme loses its optimality since it only considers IC and IA gains and the NMA scheme is due to that it is limited to one-shot linear transmission. Compared to the RGT scheme, the XTL scheme is better when $2/9<\mu_T<2/3$, because it exploits caching gains more effectively through joint ZF and IA.
Fig.~\ref{Fig Comparison_in_RAN 2}  further shows that when $\mu_T=1/3$, i.e., when the accumulated cache capacity among all the ENs is just enough to store the entire library, the performances of the XTL and HND schemes overlap and are better than the NMA and RGT schemes. Nevertheless, there is still a  non-negligible gap between the best achievable scheme and the lower bound in \eqref{eqn taul2} at small cache size regimes, including $(\mu_R=1/3, 2/9 < \mu_T < 2/3)$ and  $( 0 < \mu_R < 2/3, \mu_T = 1/3)$. This gives rise to the opportunities of tightening the lower bound further or advancing the existing delivery schemes.

%

\section{Cache-aided RAN with dedicated fronthaul}\label{s:Dedicated}

In this section, we study the joint design of cloud processing and edge caching in RANs with dedicated fronthaul links, also referred to as F-RANs, with each EN having a dedicated fronthaul link. In the F-RAN model, the ENs can fetch contents from the cloud through dedicated finite-capacity frounthaul links (see Fig. \ref{Fig model dedicated}). These dedicated links can help overcome the ENs' limited storage capacity. To this end, we introduce two transmission schemes considering caches at both the ENs and UEs, in addition to dedicated fronthaul links, where the fronthaul link capacity as well as the users' demands is unknown during the placement phase. The first scheme exploits centralized cache placement while the second one is decentralized.

\subsection{Caching for dedicated fronthaul links}\label{ss:novel dedicated}
In this section we study both centralized and decentralized cache placement at the UEs, while caching at the ENs is done in a centralized manner. We highlight that centralized coordination of the cache contents at the ENs, which model fixed base stations, is a sensible assumption. At the UE side, we first consider centralized cache placement to illustrate the main ideas, and then focus on decentralized caching, which is more appropriate to model the mobile behavior of UEs roaming around. The proposed delivery strategies for centralized and decentralized caching are based on the ideas presented in \cite{gunduz} and the soft-transfer delivery scheme in \cite{simeone}. These achievable schemes aim to minimize the NDT taking into account the interplay between the ENs' caches, UEs' caches and the capacity of the fronthaul links. The proposed delivery strategies jointly exploit cache-aided IA, ZF, and IC as well as the ENs' fronthaul links, and is studied for both half- and full-duplex transmission at the ENs.

In comparison with \cite{simeone}, where the authors consider a F-RAN with caches only at the ENs, our model also exploits caches at the UE side, similarly to \cite{girgis2017decentralized,Cran}. Moreover, we do not assume the knowledge of the capacity of the fronthaul links during the placement phase, which is a more realistic assumption, since the fronthaul link condition and its capacity can be time-varying and unknown during off-peak traffic periods.
\subsubsection{Cache Placement Phase}\label{s:Decentralized}
The ENs leverage the following centralized cache placement strategy (see Fig. \ref{placement}):
\begin{itemize}
  \item $\mu _T<\frac{1}{K_T}$: EN $p$, for $p\in [K_T]$, stores $\mu_TF$ non-overlapping bits of each file of the library, and the remainder of the files are accessible only from the edge cloud through the fronthaul links.
  \item $\mu_T\geq\frac{1}{K_T}$: Each file of the library is split into two parts, one part is replicated at all the ENs while the other part is stored collectively across the ENs (each EN caches a distinct part). As a result, each EN stores $\frac{(1-\mu _T)F}{K_T-1}$ non-overlapping bits of  each file of the library plus the same $\frac{(K_T\mu _T-1)F}{K_T-1}$ bits of each file, fulfilling the memory size constraint.
\end{itemize}
Unlike \cite{simeone}, the fronthaul link capacity is unknown during the placement phase; therefore, the placement cannot be optimized based on the fronthaul multiplexing gain $r_D$.

We consider both centralized and decentralized cache placement at the UEs. In the case of centralized cache placement, we adopt the symmetric file splitting and caching scheme of Section \ref{section cache model} at the UEs when $K_R\mu_R=i\in[0:K_R]$ is an integer.
For each file we denote the fractional size of the subfile stored at $K_R\mu _R=r$ out of $K_R$ UEs,  $r \in \mathbb{Z}$, by:
\begin{align}
\begin{split}
   f'(K_R,\mu _R)=\frac{1}{\binom{K_R}{K_R\mu _R}}.
\end{split}
\end{align}

\begin{figure}[!t]
\centering
\includegraphics[scale=0.27]{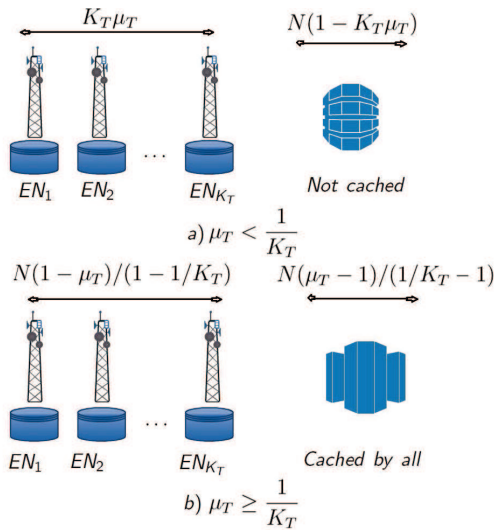}
\caption{Library partition after the cache placement at ENs.}
\label{placement}
\end{figure}
\par
For decentralized cache placement, we adopt the scheme of  Section \ref{section cache model}. We denote by $f_r$ the fractional size of the subfile stored at any $r$ out of $K_R$ UEs, each of them randomly caching $\mu _RF$ bits from each file as defined in \eqref{eqn decentralized receiver subfile size}.

\begin{remark}
We remark here that, when centralized caching is employed at both the ENs and UEs, while coordination among UE caches is needed, no coordination is required between UEs and ENs.
If the portion of each file to be stored across all the ENs is chosen randomly, for a sufficient large $F$, we can argue by the law of large numbers that, each of the subfiles stored at the UEs is also divided into two parts, one stored cooperatively across all the ENs and the other stored by either all the ENs, or none of them (depending on $\mu_T$). 
\end{remark}

\subsubsection{Delivery Phase}\label{s:delivery}


For the worst-case scenario we assume that each UE requests a distinct file from the library, and let, without loss of generality, UE $q$ request file $W_q$, $q \in [K]$.
\begin{enumerate}
  \item (IA and IC gains) We first consider the subfiles cached at one EN and $j$ UEs:
\begin{align}
\{W_{i,\Phi,\{p\}}\!:\!i\!\in
\![K_R],|\Phi|\!=\!j,\Phi\subset[K_R],i\notin\Phi,p\!\in\![K_T]\}. \label{eqn fileset 1}
\end{align}
Since each UE already has some of the undesired subfiles cached, it can cancel the interference caused by these subfiles. Therefore, by combining IC with IA, the achievable NDT for these subfiles is given by
\begin{equation} \label{eqmax}
\tau _{IA}(j)=\frac{K_R\ \binom{K_R-1}{j}f'}{\max \left\lbrace \frac{K_TK_R}{K_T+K_R-(j+1)},j+1\right\rbrace},
\end{equation}
where $f'$ denotes the fractional size of the transmitted subfiles $\{W_{i,\Phi,\{p\}}:p\in[K_T]\}$ for each $i$ and $\Phi$ in \eqref{eqn fileset 1}, which is given by $f'=f_j$ for decentralized placement and $f'=f'(K_R,\mu _R)$ for centralized placement.  In \eqref{eqmax}, the numerator represents the total fractional size of the transmitted subfiles in \eqref{eqn fileset 1}. The denominator represents the achievable sum DoF, where the first argument of the $\max$ operator corresponds to the DoF achieved in \cite{gunduz} by joint IA and IC, and the second argument corresponds to the joint transmission of the subfiles using IC. More specifially, to achieve the second argument, ENs take turn to transmit messages using XOR combining to all the UEs, and at each turn the access network becomes a single-server with a shared link as in \cite{fundamentallimits}.

\item (ZF and IC gains) Next we consider the subfiles cached at all the ENs and $j$ UEs:
\begin{align}
\{W_{i,\Phi,[K_T]}:i\in[K_R],|\Phi|=j,\Phi\subset[K_R],i\notin\Phi\}.\label{eqn fileset 2}
\end{align}
Similarly to the previous case, each UE can cancel the interference caused by undesired subfiles already cached locally. As a result, by leveraging a combination of IC and ZF techniques, the achievable NDT of  subfiles in \eqref{eqn fileset 2} is given by:
\begin{equation}\label{eqnZF}
\tau _{ZF}(j)=\frac{K_R\ \binom{K_R-1}{j}f' }{\min \left\lbrace K_T+j,K_R\right\rbrace},
\end{equation}
where $f'$ denotes the fractional size of each of the transmitted subfiles in \eqref{eqn fileset 2}, which is given by $f'=f_j$ for decentralized placement and $f'=f'(K_R,\mu _R)$ for centralized placement. Again, the numerator in \eqref{eqnZF} corresponds to the total fractional size of the transmitted subfiles in \eqref{eqn fileset 2}, while the denominator corresponds to the sum DoF. If the files to be transmitted are carefully selected, the ENs that can cache the same contents can reduce the number of interfering signals at the UEs by $j$. Consider, for example,the $3 \times 3$ F-RAN and $j=1$.  Subfiles $W_{1,\{2\},[3]}$, $W_{2,\{3\},[3]}$ and $W_{3,\{1\},[3]}$, requested by UE 1, UE 2 and UE 3, respectively, can be jointly transmitted, and by ZF we can cancel $W_{2,\{3\},[3]}$ at UE 1, $W_{3,\{1\},[3]}$ at UE 2 and $W_{1,\{2\},[3]}$ at UE 3. The interfering subfiles at each UE are already cached at this UE, e.g., $W_{3,\{1\},[3]}$ is cached at UE 1, so these interferences can be canceled. As a result, the desired subfiles are received interference-free with an equivalent DoF of $K_T+j$ over the access link.
\end{enumerate}
\par
Now, we proceed to present the delivery strategies based on the IA-IC and ZF-IC techniques. Depending on the EN cache size and the fronthaul capacity, three different delivery strategies are proposed: access-only delivery for $\mu_T\geq 1/K_T,r_D=0$, cloud-only delivery for $\mu_T=0,r_D>0$, and joint cloud and access-aided delivery for $0<\mu_T< 1/K_T,r_D>0$. These delivery strategies are based on half-duplex transmissions, while their full-duplex counterparts are obtained later.

\textbf{Access-Only Delivery ($\mu_T\geq 1/K_T,r_D=0$):} When cloud links are not available, i.e., $r_D=0$, all demands must be satisfied from the EN and UE caches as in Section \ref{s:Dedicated}.  In the proposed cache placement strategy (Section \ref{s:Decentralized}) when $\mu _T \geq 1/K_T$, each file is divided into two parts. One part is collectively cached across all the ENs with each EN caching a distinct part, while the other is replicated at all the EN caches. As a result, the transmission of the required subfiles can be carried out by a combination of the IA-IC and ZF-IC techniques, and achieves the following NDT:
\begin{equation}
\tau ^D_{a}=\tau ^D_{a_F}+\tau ^D_{a_A},
\end{equation}
where $\tau ^D_{a_F}=0$ is zero due to the lack of a fronthaul link. For $K_R\mu _R\in \mathbb{Z}$ the access NDT for centralized caching is given by
\begin{align}
\tau^D_{a_A}& = K_T\frac{1-\mu _T}{K_T-1}\tau _{IA}(K_R\mu _R)+\frac{K_T\mu_T-1}{K_T-1}\tau _{ZF}(K_R\mu _R),
\end{align}
while the non-integer points can be obtained through memory-sharing. The access NDT for decentralized caching is
\begin{align}
\tau^D_{a_A}&=\sum _{j=0}^{K_R-1}\left( K_T\frac{1-\mu _T}{K_T-1}\tau _{IA}(j)+\frac{K_T\mu_T-1}{K_T-1}\tau _{ZF}(j) \right).
\end{align}


\textbf{Cloud-Only Delivery ($\mu_T=0,r_D>0$):}  Cloud-only delivery occurs when there are no caches at the ENs, i.e.,  $\mu_T=0$, so the UEs' demands can only be satisfied by the cloud server, which requires a non-zero fronthaul link capacity, i.e., $r_D>0$. In this case, we employ the soft-transfer technique \cite{simeone} to deliver the bits of each of the requested $K_R$ files that are not already cached locally at the requesting UE. In the soft-transfer scheme the cloud server implements ZF-beamforming over the access network treating all the ENs as one virtual multi-antenna transmitter. The resulting encoded signals that should be transmitted by the ENs are quantized and transmitted to the ENs over the fronthaul links. In the soft-transfer approach, the UE caches are exploited for both the ZF and IC gains as explained above. The number of UEs at which the transmitted signal for each subfile can be neutralized or cancelled is $\min \left\lbrace K_R,K_T+j\right\rbrace-1$, exploiting the IC and ZF gain of Example 1.

For this particular network configuration, the following NDT is achievable:
\begin{equation}
\tau ^D_c=\tau ^D_{c_F}+\tau ^D_{c_A},
\end{equation}
where we have, for $K_R\mu _R\in \mathbb{Z}$,

\begin{equation*}
\begin{split}
\tau ^D_{c_F}&=\frac{K_R \binom{K_R-1}{K_R\mu_R}}{K_T r_D}f'(K_R,\mu_R),\\
\tau ^D_{c_A}&= \frac{K_R\ \binom{K_R-1}{K_R\mu_R}}{\min \left\lbrace K_R,K_T+K_R\mu_R\right\rbrace}f'(K_R,\mu_R),
\end{split}
\end{equation*}
for centralized caching, while for decentralized caching we have
\begin{equation*}
\begin{split}
\tau ^D_{c_F}&=\sum_{j=0}^{K_R-1}\frac{K_R\binom{K_R-1}{j} }{K_T r_D}f_j,\\
\tau ^D_{c_A}&= \sum _{j=0}^{K_R-1}\frac{K_R\ \binom{K_R-1}{j}}{\min \left\lbrace K_R,K_T+j\right\rbrace}f_j.
\end{split}
\end{equation*}

\textbf{Joint Cloud and Access-Aided Delivery ($0<\mu_T< 1/K_T,r_D>0$):} When $0<\mu_T< 1/K_T$, the ENs cannot store the whole library collectively; thus, both the fronthaul links and the EN caches must be used for the successful delivery of the requests. Based on the cache placement scheme in Section  \ref{s:Decentralized},  part of the requested files are available in each of the ENs, while the rest of them will be sent through the fronthaul links. The subfiles that are available  at the EN caches are transmitted using the IA-IC techniques, and the rest through the soft-transfer scheme. Therefore, the NDT achieved by centralized caching, for $K_R\mu _R\in \mathbb{Z}$, is given by
\begin{align}
\tau ^D_{h}=\tau ^{D}_{h_F}+ \tau ^D_{h_A},
\end{align}
where
\begin{align*}
\tau ^{D}_{h_F}&=(1-K_T\mu _T) \cdot \tau^D_{c_{F}},\\
\tau ^D_{h_A}&=K_T\mu _T \cdot \tau_{IA}(K_R\mu_R) +(1-K_T\mu _T) \cdot \tau ^D_{c_{A}},
\end{align*}
while the NDT of the decentralized caching scheme is
\begin{align}
\tau ^D_{h}=\sum _{j=0}^{K_R-1}K_T\mu _T \cdot \tau_{IA}(j) + (1-K_T\mu _T) \cdot \tau ^D_{c},
\end{align}
with
\begin{align*}
\tau ^{D}_{h_F}&=(1-K_T\mu _T) \cdot \tau ^D_{c_{F}},\\
\tau ^D_{h_A}&=\sum _{j=0}^{K_R-1}K_T\mu _T \cdot \tau _{IA}(j) +(1-K_T\mu _T) \cdot \tau ^D_{c_{A}}
.\end{align*}
We note that the NDT for non-integer points can be obtained by memory-sharing as before.

Combining the three delivery strategies proposed in Section \ref{s:delivery}, the following theorems provide an upper bound on the optimal NDT for \textit{half-duplex} and \textit{full-duplex} EN transmissions.
\begin{theorem}\label{thm joan 1}
For a cache-aided F-RAN with $K_T\ge2$ ENs, each with a cache of normalized size $\mu_T$,  $K_R\ge2$ UEs, each with a cache of normalized size $\mu_R$,  $N\ge K_R$ files, and a dedicated fronthaul link with capacity $C_F=r_D\log P>0$, the following NDT can be achieved by \textbf{half-duplex transmission}
\begin{equation}
\tau _S ^D=
\begin{cases}
\min \{ \tau ^D_{h}, \tau^D_{c}\},&\text{ if }\mu _T< \frac{1}{K_T}\\
\min \{ \tau ^D_{a}, \tau^D_{c}\},&\text{ if }\mu _T\geq \frac{1}{K_T}\\
\end{cases}.
\end{equation}
\end{theorem}

\begin{proof}
In half-duplex transmission, the total NDT is the sum of the fronthaul  ($\tau_F$) and access ($\tau_A$) NDTs, which corresponds to the minimum of the NDTs of the \textit{cloud-only delivery} or \textit{joint cloud and access-aided delivery} when $\mu _T<1/K_T$; and the minimum of the NDTs of the \textit{cloud-only delivery} or \textit{access-only delivery} when $\mu _T\geq 1/K_T$.  Once the cloud link capacity is revealed, the best transmission scheme is chosen based on the fronthaul link rate $r_D$ and the EN cache size $\mu_T$. If $r_D$ is small, e.g., high network congestion,  \textit{joint cloud and access-aided} delivery will be used if $\mu _T<1/K_T$, and \textit{access-only} delivery if $\mu _T\geq 1/K_T$. On the other hand, if $r_D$ is large, \textit{cloud-only} approach outperforms the other two schemes.
\end{proof}

\begin{theorem}\label{thm joan 2}
For the cache-aided F-RAN with $K_T\ge2$ ENs, each with a cache of normalized size $\mu_T$,  $K_R\ge2$ UEs, each with a cache of normalized size $\mu_R$,  $N\ge K_R$ files, and a dedicated fronthaul link with capacity $C_F=r_D\log P>0$, the following NDT can be achieved by \textbf{full-duplex transmission}
\begin{equation}
\tau_P ^D\!=\!
\begin{cases}
\min \{\max \{ \tau ^{D}_{h_F}, \tau ^{D}_{h_A}\},\max \{\tau ^{D}_{c_F}, \tau ^{D}_{c_A}\}\},&\text{ if }\mu _T\!\leq\! \frac{1}{K_T}\\
\min \{\max \{ \tau ^{D}_{c_F}, \tau ^{D}_{c_A}\},\tau ^{D}_{a_A}\},&\text{ if }\mu _T\!\geq\! \frac{1}{K_T}
\end{cases}.
\end{equation}
\end{theorem}

\begin{proof}
From the results in \cite{simeone} for this type of transmission, we only need to prove the achievability of the fronthaul and access NDTs, which follow from Theorem \ref{thm joan 1}.
\end{proof}

\subsection{Numerical results}
In this subsection, we will present the NDT achieved by the caching and delivery schemes presented above, for some particular network setting, and compare the results with other schemes available in the literature. We first briefly introduce the benchmark schemes from the literature.
\subsubsection{Fully Centralized Caching}

In \cite{simeone}, the authors assume that only the ENs are equipped with caching capabilities. In the delivery phase, by exploiting  ZF (for subfiles cached at all the ENs), IA (for subfiles cached at only one EN), and soft-transfer (for subfiles not cached at any of the ENs) techniques, the authors show that the following  NDT is achievable via half-duplex EN transmission:
\begin{align}
  \tau_{\textrm{STS}}=\left\{\!
  \begin{array}{ll}
    (K_T\!+\!K_R\!-\!1)\mu_T+(1\!-\!\mu_TK_T)\left(\frac{K_R}{\min\{K_T,K_R\}}\!+\!\frac{K_R}{K_Tr_D}\right),&\textrm{ for $\mu_T\le\frac{1}{K_T},r_D\le r_{th}$},\\
    \frac{K_R}{\min\{K_T,K_R\}}\frac{K_T\mu_T-1}{K_T-1}+(1-\mu_T)\frac{K_T+K_R-1}{K_T-1},&\textrm{ for $\mu_T\ge\frac{1}{K_T},r_D\le r_{th}$},\\
    \frac{K_R}{\min\{K_T,K_R\}}+\frac{(1-\mu_T)K_R}{K_Tr_D},&\textrm{ for $0\le\mu_T\le1,r_D\ge r_{th}$}.
  \end{array}
  \right.
\end{align}
where $r_{th}\triangleq \frac{K_R(K_T-1)}{K_T(\min\{K_T,K_R\}-1)}$.

In \cite{goselingdelivery}, the authors generalize this result by introducing the achievable \textit{NDT region} to characterize the trade-off among the latencies achieved by different users' demand combinations. An achievable scheme is presented for a F-RAN with two ENs and two UEs.

\subsubsection{Fully Decentralized Caching}
The authors in \cite{girgis2017decentralized} consider decentralized cache placement at both the EN and the UE sides. Note that, as discussed in Section \ref{section cache model}, with decentralized caching at the ENs the presence of fronthaul links is a requirement to satisfy all possible UE demands. In \cite{girgis2017decentralized} the authors propose a delivery scheme for F-RAN with two ENs, which leverages ZF, IA and soft-transfer techniques opportunistically. The achievable NDT via half-duplex EN transmission is given by
\begin{align}
  \tau_{\textrm{GENE}} =
  \left\{
  \begin{array}{ll}
  \tau_F^a+\tau_A^a,&\textrm{ for $0<r_D\le K_R$}\\
  \tau_F^b+\tau_A^b,&\textrm{ for $K_R < r_D$}
  \end{array}
  \right.,\notag
\end{align}
where
\begin{align}
  &\tau_F^a\triangleq \frac{(1-\mu_T)^2(1-\mu_R)}{r_D\mu_R}\cdot\left[1-(1-\mu_R)^{K_R}-\frac{K_R\mu_R}{2}(1-\mu_R)^{K_R-1}\right],\notag\\
  &\tau_F^b\triangleq \frac{(1-\mu_T)^2(1-\mu_R)}{r_D\mu_R}\cdot\left[1-(1-\mu_R)^{K_R}-\frac{K_R\mu_R}{2}(1-\mu_R)^{K_R-1}\frac{1-3\mu_T}{1-\mu_T}\right],\notag\\
  &\tau_A^a\triangleq \frac{1-\mu_R}{\mu_R}\cdot\bigg[1-(1-\mu_R)^{K_R}-\left(\frac{K_R}{2}-\mu_T(1-\mu_T)\right)\mu_R(1-\mu_R)^{K_R-1}\bigg],\notag\\
  &\tau_A^b\triangleq \frac{1-\mu_R}{\mu_R}\cdot\left[1-(1-\mu_R)^{K_R}-\frac{K_R\mu_R}{2}(1-\mu_R)^{K_R-1}\right].
\end{align}
\subsubsection{Numerical Comparison}
 In what follows, we present the comparison between the achievable NDTs of the proposed caching and delivery schemes, the scheme presented in \cite{simeone} (referred to as STS), and the one in \cite{girgis2017decentralized} (referred to as GENE).

\begin{figure}[!tbp]
\centering
\includegraphics[scale=0.4]{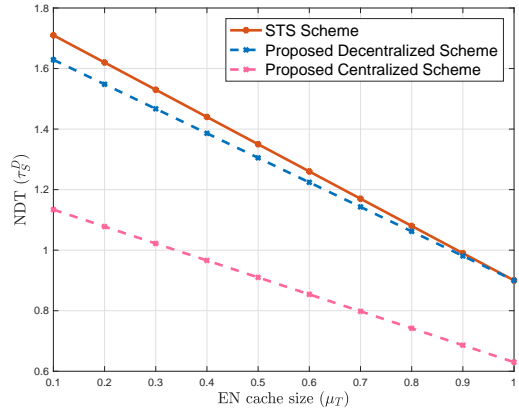}
\caption{NDT vs. EN cache size $\mu_T$ for access-only delivery when $K_T=K_R=10$,  $\mu_R=0.1$  and $r_D=0$.}
\label{no-cloud}
\end{figure}

We first consider access-only delivery, i.e., $r_D=0$, by assuming $\mu _T\geq 1/K_T$. In Fig. \ref{no-cloud}, we compare the NDT of the proposed schemes with STS scheme when $K_T=K_R=10,\mu_R=0.1$. For fairness of the comparison, we added local caching gain to the STS scheme, since it originally only considers EN caches. Fig. \ref{no-cloud} illustrates the gains from UE caches in terms of the NDT in a F-RAN. We observe that as the EN cache size increases, the performance improvement of the proposed schemes compared to STS shrink. This is because, as $\mu_T$ increases, the number of subfiles transmitted using ZF in our delivery scheme increases, and the benefit of UE caches for IC diminishes, as they only account for uncoded caching gain as in the STS scheme. However, for a limited $\mu_T$, UE caches can provide gains beyond uncoded caching gain thanks to the combination of IA, ZF and IC techniques. Moreover, it can be seen that centralized and decentralized schemes are very close in performance; therefore, even when the UE caches cannot be centrally coordinated, the loss in NDT is relatively small.

\begin{figure}[!tbp]
\centering
\includegraphics[scale=0.4]{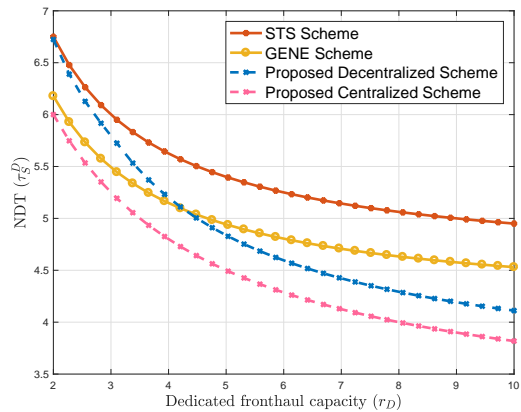}
\caption{NDT vs. founthaul link capacity $r_D$ for cloud-only delivery when $K_T=2$, $K_R=10$, $\mu _R=0.1$ and $\mu _T=0$.}
\label{no-cahce}
\end{figure}

In Fig. \ref{no-cahce}, we consider cloud-only delivery, i.e., $\mu _T=0$, with half-duplex EN transmission. Here, we plot the NDT performance with respect to the fronthaul link capacity $r_D$. As expected, the NDT decays with $r_D$, and saturates to a fixed value, which essentially characterizes the access delay.  We consider $K_T=2$ to be able to compare the results with that of the GENE scheme of \cite{girgis2017decentralized}. It must be noted that the STS scheme only exploits local caching gain from UE caches, while the GENE scheme assumes decentralized caching at the ENs; and hence, their performance is relatively poorer. The GENE scheme has worse performance compared to our proposed decentralized scheme when $r_D$ is high, because the former scheme employs soft-transfer scheme only for a part of the files that is not cached anywhere in the network, whereas our proposed decentralized scheme employs soft-transfer scheme that enables ZF at the ENs and also benefits from the UE caches. As the cloud rate increases, the benefit of joint soft-transfer and centralized cache placement outperforms significantly the GENE scheme.

\begin{figure}[!tbp]
\centering
\includegraphics[scale=0.4]{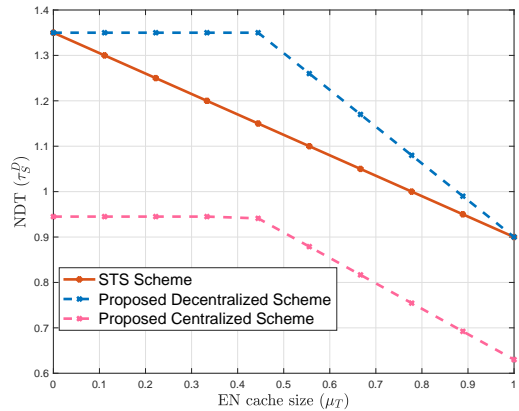}
\caption{NDT vs. EN cache size $\mu_T$ for joint cloud and access-aided delivery when $K_T=K_R=10$, $\mu _R=0.1$, $r_D=2$.}
\label{general}
\end{figure}

Joint cloud and access-aided delivery is considered in Fig. \ref{general}.  We observe that the performance of the proposed centralized scheme is significantly better than that of the STS scheme, thanks to the coordination of the UE caches, and to the exploitation of  IA, ZF and IC techniques jointly. We reemphasize that our caching strategies do not assume the knowledge of the fronthaul link capacities. This is motivated from the practical consideration that the placement and delivery phases are typically carried out over different time frames, and an accurate prediction of the fronthaul link capacities during the placement phase is too strong an assumption. The consequence of this limitation can be observed in Fig. \ref{general} where, due to the high cloud link capacity, the STS scheme achieves a lower NDT compared to our proposed decentralized scheme. The initial flat performance of the proposed schemes is because we do not start exploiting the EN caches until $\mu_T=0.4$, and employ the soft-transfer scheme before that point, whose performance does not depend on $\mu_T$ in this case since we assume $K_T = K_R$. However, even though the cloud rate is unknown during the placement phase, the proposed decentralized scheme performance approaches that of  the STS as $\mu_T$ increases.

\begin{figure}[!tbp]
\centering
\includegraphics[scale=0.4]{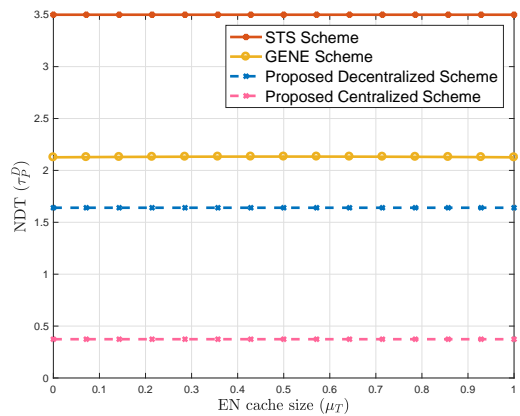}
\caption{NDT vs. EN cache size $\mu_T$ for joint cloud and access-aided delivery with full-duplex EN transmission when $K_T=2$, $K_R=10$, $\mu _R=0.1$, $r_D=3$.}
\label{pipe1}
\end{figure}
\par

We conclude this section by considering  full-duplex ENs. As expected, Fig. \ref{pipe1} shows significant reduction in the achieved NDT compared to half-duplex ENs. Particularly noticeable is the the low NDT achieved by the centralized placement scheme. Note in Figure \ref{general} that the NDT of 0.95 is achievable in the case of half-duplex ENs for the same setting with cloud only delivery (i.e. $\mu_T=0$), while it reduces to $0.5$ for full duplex ENs.
\par

Interesting is the NDT behavior of the pipeline transmission presented in Figure \ref{pipe1}, where a clear reduction of it is obtained. For the configuration presented in the figure, the lowest NDT is always obtained by soft-transfer delivery.  For  $\mu_T = 0, 0.1, .0.2$  the maximum between edge and the fronthaul NDT is initially given by the fronthaul. Then as $\mu_T$ increases the fronthaul delay is reduced and the edge increases (down slope). Then IA at the edge becomes more of a burden and the edge delay becomes dominant increasing until $\mu_T$ = 0.5. Note that, contrary to what is usually believed in caching, memory sharing could not be employed as we do not know the cloud rate in advance. Memory sharing relies on that all network parameters are known, and thus the line between two achievable NDT points is also achievable  by means of splitting the caches in two parts, one that leverages the delivery scheme of one of the points and the other, that of the other point proportionally. However, in our scheme we cannot modify the cache placement beforehand as  $r_D$ is unknown during the placement phase and thus, we cannot obtain the line of achievable points \textit{a priori.} As a result, the proposed scheme becomes one possible solution which tries to satisfy the worst case, i.e., lack of cloud links.

\section{Cache-Aided RAN with Wireless Fronthaul}
\label{section Caching gain in F-RAN with wireless fronthaul}
In this section, we consider the F-RAN where each EN is connected to the cloud via a shared wireless link operated by an MBS as shown in Fig.~\ref{Fig model wireless}.   It is important to emphasize that the fronthaul network from the cloud to the ENs can be viewed as a broadcast channel with receiver caches, where coded multicasting or IC can be exploited but needs to be designed jointly with the access network in the F-RAN model.
The coded caching framework in F-RAN with wireless fronthaul has been previously studied in   \cite{koh2017cloud,Cran}.
In specific, \cite{koh2017cloud} focused on a $2\times2$ network with EN caches and proposed a scheme that exploits both coded and uncoded multicasting over the fronthaul link as well as IA and ZF opportunistically over the access link. It is shown that, under full-duplex EN transmission, coded multicasting over the fronthaul link is unnecessary to achieve the optimal NDT performance in certain cases.
The work \cite{Cran} considered a general $K_T\times K_R$ F-RAN architecture with caches at both the EN and UE sides, and presented a network-coded fronthauling strategy in conjunction with ZF over the access link for half-duplex ENs.
Note that both \cite{koh2017cloud} and \cite{Cran} assumed centralized cache placement.
In this section, we propose a new delivery scheme with decentralized cache placement at all the ENs and UEs for this model.  Note that, due to the wireless fronthaul connection, considering decentralized cache placement at the EN side as well is of great practical interest.

For simplicity, we focus on half-duplex EN transmission only,  though the extension to full-duplex transmission is straightforward.  In our proposed scheme,  the wireless fronthaul link is used not only to fetch the requested file bits which are not available in any EN cache, but also the file bits already cached at some but not all ENs to boost transmission cooperation to any desired level in the access link.
%
%
The access transmission in our proposed delivery scheme is similar to \cite{mine}, which transforms the access link into a cooperative X-multicast channel. Based on the proposed delivery scheme, we obtain an achievable upper bound on the optimal NDT.  We also obtain a theoretical lower bound on the optimal NDT following cut-set-like arguments in the fronthaul and access networks separately. It is shown that the multiplicative gap between the upper and lower bounds is within $12$.

\subsection{Delivery Scheme}\label{section V 3x3}
In this subsection we present the proposed delivery scheme using the $3\times3$ F-RAN model as an example.
%
As before, we assume that UE $q$  desires $W_q$, for $q\in[3]$. Using the notations defined in Section \ref{section cache model}, we denote $W_{q,\Phi,\Psi}$ as the subfile desired by UE $q$ and cached at UE set $\Phi$ and EN set $\Psi$.
Excluding the locally cached subfiles,  each UE $q$, for $q\in[3]$, wants to receive subfiles $\{W_{q,\Phi,\Psi}:\Phi\not\ni q,\Phi\subseteq[3],\Psi\subseteq[3]\}$. We divide the subfiles wanted by all the UEs into different groups according to the size of $\Phi$ and $\Psi$, indexed by $\{(m,n):m\in[0:2], n\in[0:3]$, such that subfiles in group $(m,n)$ are cached at $m=|\Phi|$ out of $K_R$ UEs and $n=|\Psi|$  out of $K_T$ ENs.  As a result of random decentralized cache placement,  the fractional size of each subfile in group $(m,n)$ is given by $f_{m,n}$ as shown in \eqref{eqn decentralized subfile size} at large file size.
There are $3\binom{2}{m}\binom{3}{n}$ subfiles in group $(m,n)$. Each group of subfiles is delivered individually in a time-division manner. In the following, we present the delivery strategy for two representative groups, $(m,0)$ and $(m,1)$, where $m\in[0:2]$.

\subsubsection{Delivery of Group $(m,0)$}\label{section delivery33 m0}
Each subfile in group $(m,0)$ is desired by one UE, cached at $m$ other UEs, but none of the ENs. Fronthaul transmission is compulsory in the delivery phase since these subfiles are not available at any of the ENs. Instead of transmitting these subfiles to all the ENs one by one over the wireless fronthaul link, we utilize the local cache contents of the UEs, if $m\ne 0$,  and exploit the coded multicasting gain through XOR combining for these subfiles, similarly to \cite{fundamentallimits}. The specific delivery scheme is given below.

\textbf{fronthaul-compulsory delivery}:  The cloud generates a set of coded messages given by
\begin{align}
  \left\{ W_{\Phi^+,\emptyset}^\oplus \triangleq  \bigoplus_{q\in\Phi^+}  W_{q,\Phi^+\backslash \{q\},\emptyset} : \Phi^+ \subseteq [3],|\Phi^+| = m+1 \right\},\label{eqn example m0 1}
\end{align}
where each coded message  $W_{\Phi^+,\emptyset}^\oplus$  has $f_{m,0}F$ bits, and is desired by the UEs in set $\Phi^+$.  We let the MBS naively multicast each coded message in \eqref{eqn example m0 1}  to all the three ENs. The fronthaul NDT is then given by
\begin{align}
  \tau_F=\frac{\binom{3}{m+1}f_{m,0}}{r_W},\label{eqn 33m0 fronthaul}
\end{align}
where  $r_W$ is the multiplexing gain of the wireless fronthaul link defined before.

By such naive multicasting in the fronthaul link, each EN now has access to all the coded messages in \eqref{eqn example m0 1}, and can transmit with full cooperation in the access network. The access channel thus becomes the $\binom{3}{3}\times\binom{3}{m+1}$ cooperative X-multicast channel, whose achievable per-user DoF is $d_{m,3}=1$ by \eqref{eqn tau ip}. Since each UE desires $\binom{2}{m}$ coded messages, the access NDT is given by
\begin{align}
  \tau_A=\binom{2}{m}f_{m,0}.\label{eqn 33m0 access}
\end{align}

Summing up \eqref{eqn 33m0 fronthaul} and \eqref{eqn 33m0 access}, the total NDT for group $(m,0)$ is
\begin{align}
  \tau_{m,0}=\frac{\binom{3}{m+1}f_{m,0}}{r_W}+\binom{2}{m}f_{m,0}.\notag
\end{align}

\subsubsection{Delivery of Group $(m,1)$}\label{section delivery33 m1}
Unlike the subfiles in group $(m,0)$, each subfile in group $(m,1)$ is already cached at one EN, and therefore the fronthaul transmission is optional. To utilize the UE caches, if $m\ne 0$, we can still generate coded messages as in \eqref{eqn example m0 1} but at each EN rather than the cloud. In specific, each EN $p$, for $p\in[3]$, generates:
\begin{align}
   \left\{ W_{\Phi^+ ,\{p\}}^\oplus \!\triangleq\!  \bigoplus_{q\in\Phi^+}\!  W_{q,\Phi^+\backslash \{q\},\{p\}} : \Phi^+ \subseteq [3], |\Phi^+ | = m + 1 \right\} . \label{eqn example m1 1}
\end{align}
Each coded message $W_{\Phi^+,\{p\}}^\oplus$ has $f_{m,1}F$ bits, and is desired by the UEs in set $\Phi^+$. In the following, we introduce the transmission of these coded messages over the access network with and without the aid of the fronthaul network, respectively.

\textbf{Access-Only Delivery:} Each EN $p$, for $p\in[3]$, sends $\{W_{\Phi^+,\{p\}}^\oplus\}$ in the access network, and the access channel becomes the $\binom{3}{1}\times\binom{3}{m+1}$ cooperative X-multicast channel with achievable per-user DoF $d_{m,1}$ given in \eqref{eqn tau ip} . Since each UE desires $3\binom{2}{m}$ messages, the  total NDT is given by
\begin{align}
  \tau=\frac{3\binom{2}{m}f_{m,1}}{d_{m,1}}.\label{eqn 33m1 nofronthaul}
\end{align}

\textbf{Fronthaul-Aided Delivery: }With the aid of fronthaul, we can allow ENs to share the coded messages in the fronthaul network so as to form transmission cooperation among ENs in the access network. As a price to pay for the EN cooperation gain, additional fronthaul delivery latency will be caused.  Thus, the optimal cooperation strategy should balance the transmission time between the access network and the fronthaul network.

Assume that after the aid of fronthaul transmission, every set of $1+i$ ENs can form a cooperation group in the access link, where $i \in  [2]$ is a design parameter to balance the tradeoff mentioned above. We split each coded message $W_{\Phi^+,\{p\}}^\oplus$ in \eqref{eqn example m1 1} into $\binom{2}{i}$ sub-messages $\{W_{\Phi^+,\{p\}}^{\oplus,\Psi^+}:\Psi^+\subseteq[3],|\Psi^+|=1+i,p\in\Psi^+\}$, each with $\frac{f_{m,1}}{\binom{2}{i}}F$ bits and sent by EN set $\Psi^+$ exclusively in the access network. Consider an arbitrary EN set $\Psi^+$ with size $1+i$. The sub-messages to be sent by this set are
\begin{align}
\left\{W_{\Phi^+,\{p\}}^{\oplus,\Psi^+}:\Phi^+\subseteq[3],|\Phi^+|=m+1,p\in\Psi^+\right\}.\label{eqn 33m1 fronthaul 1}
\end{align}
Therefore, to have these sub-messages ready at their corresponding EN sets, the MBS needs to send $\{W_{\Phi^+,\{p\}}^{\oplus,\Psi^+}:\Phi^+\subseteq[3],|\Phi^+|=m+1\}$ to ENs $\{p':p'\in\Psi^+\backslash \{p\}\}$ which do not cache them. Given that each sub-message is already cached at one EN, an additional layer of pair-wise XOR combining on the top of these  sub-messages can be used to exploit the IC gain in the fronthaul network. In specific, the MBS generates a set of coded sub-messages
\begin{align}
  \left\{ W_{\Phi^+,\{p\}}^{\oplus,\Psi^+}\oplus W_{\Phi^+,\{p'\}}^{\oplus,\Psi^+} : \Phi^+\subseteq\![3],|\Phi^+|=m\!+\!1,p,p'\!\in\Psi^+\right\}\notag
\end{align}
with each intended to ENs $p$ and $p'$ in set $\Psi^+$.
Upon receiving the above  coded sub-messages, each EN in $\Psi^+$ can decode its desired sub-messages with  its local cache. The fronthaul NDT for the given $i$ is  thus given by
\begin{align}
  \tau_F=\frac{\binom{3}{m+1}\binom{3}{1+i}\binom{1+i}{2}f_{m,1}}{r_W\binom{2}{i}}=\frac{3\binom{3}{m+1}if_{m,1}}{2r_W}.\label{eqn 33m1 fronthaul 2}
\end{align}

In the access network, the $1+i$ ENs in each set $\Psi^+$ cooperatively transmit sub-messages in \eqref{eqn 33m1 fronthaul 1}, each desired by $m+1$ UEs. The access network is thus upgraded to the $\binom{3}{1+i}\times\binom{3}{m+1}$ cooperative X-multicast channel with achievable per-user DoF $d_{m,1+i}$  in \eqref{eqn tau ip}. Since each UE wants $\binom{2}{m}\binom{3}{1+i}\binom{1+i}{1}$ sub-messages, each with $\frac{f_{m,1}}{\binom{2}{i}}F$ bits, the access NDT is
\begin{align}
  \tau_A=\frac{\binom{2}{m}\binom{3}{1+i}\binom{1+i}{1}f_{m,1}}{\binom{2}{i}d_{m,1+i}}=\frac{3\binom{2}{m}f_{m,1}}{d_{m,1+i}}.\label{eqn 33m1 fronthaul 3}
\end{align}
Summing up \eqref{eqn 33m1 fronthaul 2} and \eqref{eqn 33m1 fronthaul 3}, the total NDT is given by
\begin{align}
  \tau=\frac{3\binom{3}{m+1}if_{m,1}}{2r_W}+\frac{3\binom{2}{m}f_{m,1}}{d_{m,1+i}}.\label{eqn 33m1 fronthaul 4}
\end{align}

Finally, comparing the NDT achieved with access-only delivery in \eqref{eqn 33m1 nofronthaul} and the NDT with fronthaul-aided delivery in  \eqref{eqn 33m1 fronthaul 4}  for all possible $i$, we choose the smallest one to be the NDT for group $(m,1)$, i.e.,
\begin{align}
  \tau_{m,1}=\min_{i\in[0:2]} \left\{\frac{3\binom{3}{m+1}if_{m,1}}{2r_W}+\frac{3\binom{2}{m}f_{m,1}}{d_{m,1+i}}\right\}.\notag
\end{align}


\subsection{Main Results}
Generalizing the above delivery strategy to the $K_T\times K_R$ F-RAN with wireless fronthaul, we obtain the achievable NDT in the following theorem whose proof is given in Appendix A.

\begin{theorem}\label{thm 1}
For the cache-aided F-RAN with $K_T\ge2$ ENs, each with a cache of normalized size $\mu_T$,  $K_R\ge2$ UEs, each with a cache with normalized size $\mu_R$,  $N\ge K_R$ files, and a wireless fronthaul link with MBS power $P_F =\left( P\right)^{r_W}$, the minimum NDT achieved by  random decentralized caching with half-duplex transmission is upper bounded by
\begin{align}
\tau^W_{upper}=\sum_{m=0}^{K_R-1}\sum_{n=0}^{K_T}\tau_{m,n},\label{tau-upper}
\end{align}
where
\begin{numcases}{\tau_{m,n}\!=\!\!}\!
\binom{K_R}{m\!+\!1}\!\frac{f_{m,0}}{r_W}\!+\!\binom{K_R\!-\!1}{m}\!\frac{f_{m,0}}{d_{m,K_T}},\textrm{ if $n\!=\!0$,} \label{eqn tau m0}\\
\!\min_{i \in [0:K_T-n]} \tau_{m,n}^i, \qquad\qquad\quad\qquad\ \textrm{ if $n\!\ge\!1$,} \label{eqn tau mn gamma}
\end{numcases}
with
\begin{align}
  \tau_{m,n}^i=\binom{K_R}{m+1}\binom{K_T}{n}\min\left\{1,\frac{i}{n+1}\right\}\frac{f_{m,n}}{r_W}+\binom{K_R-1}{m}\binom{K_T}{n}\frac{f_{m,n}}{d_{m,n+i}}.\label{eqn thm 1}
\end{align}
Here $f_{m,n}$ is the fractional size of each subfile cached in $m$ UEs and $n$ ENs  given in \eqref{eqn decentralized subfile size}, and $d_{m,j}$ is the achievable per-user DoF of the $\binom{K_T}{j}\times\binom{K_R}{m+1}$ cooperative X-multicast channel given in \eqref{eqn tau ip}.
\end{theorem}

In Theorem \ref{thm 1} , \eqref{eqn tau m0} is the delivery time of those subfiles that are not cached in any EN, i.e., $n=0$, by the fronthaul-compulsory delivery scheme. It has an explicit expression, consisting of the fronthaul part and the access part.  \eqref{eqn tau mn gamma} is the delivery time of subfiles cached in at least one EN, i.e., $n\ge 1$, which takes the minimum achieved among all possible fronthaul-aided strategies indexed by $i$  in. The index $i\in [0: K_T-n]$ represents the increased level for EN cooperation over the access network boosted by the fronthaul-aided delivery.  By taking a closer look at $\tau_{m,n}^i$ in \eqref{eqn thm 1}, the first term is due to the fronthaul transmission, if $i>0$, via direct multicasting at $i> n+1$ or coded multicasting at $i\le n+1$,  while the second term is due to the access transmission over the upgraded $\binom{K_T}{n+i}\times\binom{K_R}{m+1}$ cooperative X-multicast channel where the benefits of ZF and IA are jointly exploited.  This theorem indicates that the proposed delivery scheme will always find a balance between the increased fronthaul latency (due to more bits to fetch) and the reduced access latency (due to more chance for EN cooperation), towards a minimum end-to-end latency.

%
In the extreme case when the fronthaul capacity is sufficiently large, i.e.,  $r_W \rightarrow \infty$, the fronthaul NDT approaches zero, and the overall achievable NDT is dominated by the access NDT, given by
\begin{align}
  \tau_{upper} (r_W \rightarrow \infty ) =\sum_{m=0}^{K_R-1}\sum_{n=0}^{K_T}\binom{K_R-1}{m}\binom{K_T}{n}\frac{f_{m,n}}{d_{m,K_T}}=\sum_{m=0}^{K_R-1}\binom{K_R-1}{m}\frac{\mu_R^m(1-\mu_R)^{K_R-m}}{d_{m,K_T}},\label{tau_u_infinity}
\end{align}
which is equivalent to the NDT when $\mu_T=1$.

We next obtain a lower bound of the minimum NDT based on the assumption of random decentralized cache placement, whose proof is in Appendix B.
\begin{theorem}\label{thm 2}
For the cache-aided F-RAN with $K_T\ge2$ ENs, each with a cache of normalized size $\mu_T$,  $K_R\ge2$ UEs, each with a cache with normalized size $\mu_R$,  $N\ge K_R$ files, and  a wireless fronthaul link with MBS power $P_F =\left( P\right)^{r_W}$, the minimum NDT achieved by random decentralized caching with half-duplex EN transmission is lower bounded by
\begin{align}
  \tau^W_{lower}=&\max_{l_1\in[K_R]}\frac{l_1(1-\mu_T)^{K_T}(1-\mu_R)^{l_1}}{r_W}+\max_{l_2\in[K_R]}\frac{l_2(1-\mu_R)^{l_2}}{\min\{l_2,K_T\}}.\label{eqn thm 2}
\end{align}
\end{theorem}


Comparing Theorem \ref{thm 1} and Theorem \ref{thm 2}, we obtain the multiplicative gap between the upper and lower bounds below, with proof given in Appendix C.

\begin{corollary}  \label{coro gap}
The multiplicative gap between the NDT upper bound \eqref{tau-upper} and the NDT lower bound \eqref{eqn thm 2} is within 12.
\end{corollary}

\subsection{Numerical Examples and Comparison}
First, we observe the sum NDT, the fronthaul NDT, and the access NDT achieved by the proposed delivery scheme, seperately, at different fronthual multiplexing gain $r_W$.
Fig. \ref{Fig phaseNDT} depicts the achievable NDT results as well as the NDT lower bound in the $3\times 3$ network with $\mu_T=\mu_R=\frac{1}{3}$. It can be seen that the sum NDT  decreases as  $r_W$ increases, and is very close to the lower bound in the entire region of $r_W$. When $r_W$ is large, the fronthaul NDT can be ignored, and the sum NDT converges to a limit, which, based on  \eqref{tau_u_infinity}, is only related to UE cache size $\mu_R$ but not EN cache size $\mu_T$. Comparing to the sum NDT and fronthaul NDT, the access NDT  remains constant when $r_W\le3$, then decreases very slowly when $r_W$ increases. This is because the access NDT term of subfiles in group $(m,0)$ in  \eqref{eqn tau m0} is only a function of cache sizes $\mu_R$ and $\mu_T$, and the access NDT term of subfiles in group $(m,n)$ in \eqref{eqn tau mn gamma} , for $n\ge1$, remains constant when the optimal $i$ is fixed for a range of $r_W$, and hence decreases slowly as $r_W$ increases. When $r_W=3.5$, there is a slight increase of the fronthaul NDT. This indicates that  more subfiles are transmitted in the fronthaul link at this point of fronthaul multiplexing gain to trade for higher EN cooperation in the access link.

\begin{figure}[tbp]
\begin{centering}
\includegraphics[scale=0.38]{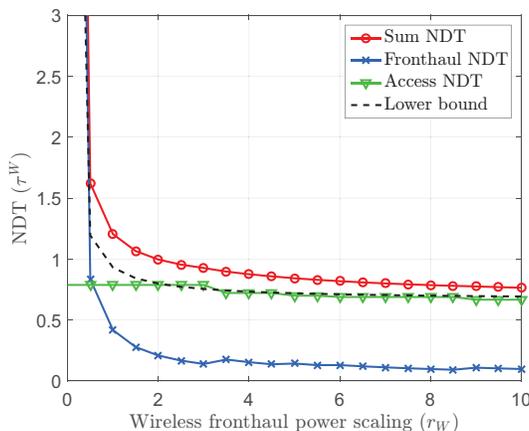}
\vspace{-0.2cm}
\caption{NDT v.s. fronthaul capacity $r_W$ when $K_T=3,K_R=3,\mu_T=\frac{1}{3},\mu_R=\frac{1}{3}$.}\label{Fig phaseNDT}
\end{centering}
\end{figure}

Next we compare our scheme with \cite{koh2017cloud,Cran} which study centralized caching in F-RAN with wireless fronthaul. The work \cite{koh2017cloud} (referred to as \emph{KSTK}) is limited to a $2\times2$ F-RAN with EN caches only and it obtains the achievable NDT for full-duplex EN, given by
\begin{align} \label{eqn:tau_KSTK}
  \tau_{\rm{KSTK}}=\left\{\!
\begin{array}{ll}
\mu_T(3-\frac{4}{r_W})+\frac{2}{r_W},&\textrm{ for $0\le\mu_T<\frac{1}{2},0<r_W\le\frac{2}{3}$},\\
2-\mu_T,&\textrm{ for $\frac{1}{2}\le\mu_T\le1, 0<r_W\le\frac{2}{3}$},\\
(1-\mu_T)\frac{2}{r_W},&\textrm{ for $0\le\mu_T<\frac{1}{2}, \frac{2}{3}<r_W\le1$},\\
2\mu_T(1-\frac{1}{r_W})+\frac{2}{r_W}-1,&\textrm{ for $\frac{1}{2}\le\mu_T\le1, \frac{2}{3}<r_W\le1$},\\
2(1-\frac{2}{r_W})\mu_T+\frac{2}{r_W},&\textrm{ for $0\le\mu_T<\frac{1}{2}, 1<r_W<2$},\\
1,&\textrm{ for $\frac{1}{2}\le\mu_T\le1, 1<r_W<2$,}\\
1,&\textrm{ for $0\le\mu_T\le1,r_W\ge 2$}.
\end{array}
\right..
\end{align}
%
The work \cite{Cran} (referred to as \emph{DYL}) considers a general $K_T\times K_R$ F-RAN architecture with caches at both the ENs and UEs and obtains an achievable  NDT under half-duplex EN transmission, given by:
\begin{align}
\tau_{\rm{DYL}}=\left\{
\begin{array}{ll}
\frac{K_R}{r_W}+\frac{K_R-K_R\mu_R}{\min\{K_T+K_R\mu_R,K_R\}},&\textrm{ if $\mu_T=0$ and $K_R\mu_R \in \mathbb{Z}$}\\
\min\limits_{t\in[K_T\mu_T:K_T]}\{\frac{(t-K_T\mu_T)K_R}{tr_W}+\frac{K_R-K_R\mu_R}{\min\{t+K_R\mu_R,K_R\}}\},&\textrm{ if $K_T\mu_T\in \mathbb{Z}^+$,$K_R\mu_R \in \mathbb{Z}$}
\end{array}
\right. .\label{eqn:tau_DYL}
\end{align}



\begin{figure}[!tbp]
\begin{minipage}[t]{1\linewidth}
\centering
\includegraphics[scale=0.4]{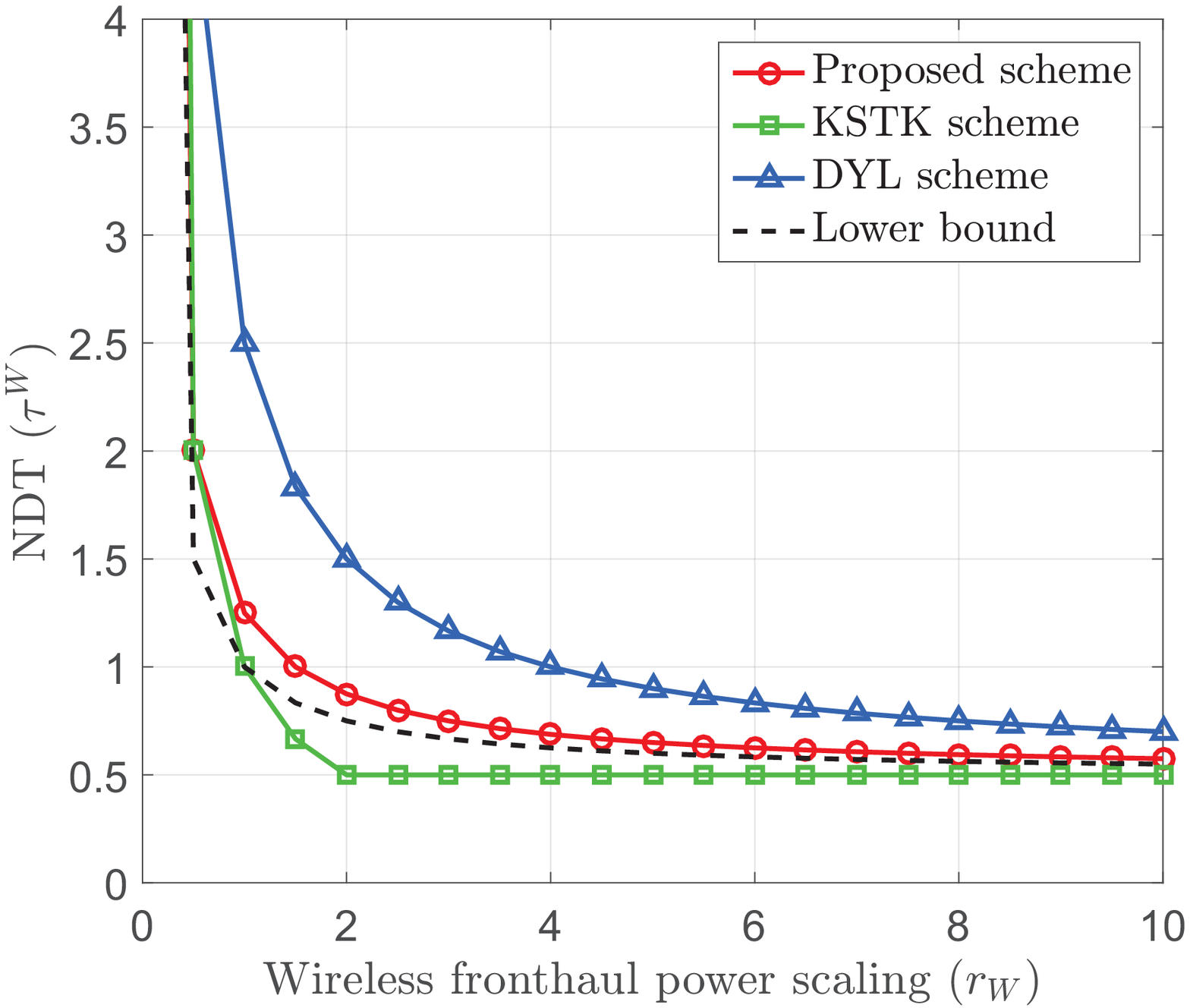}
\subcaption{}\label{Fig compare2 mut0}
\end{minipage}
\begin{minipage}[t]{1\linewidth}
\centering
\includegraphics[scale=0.4]{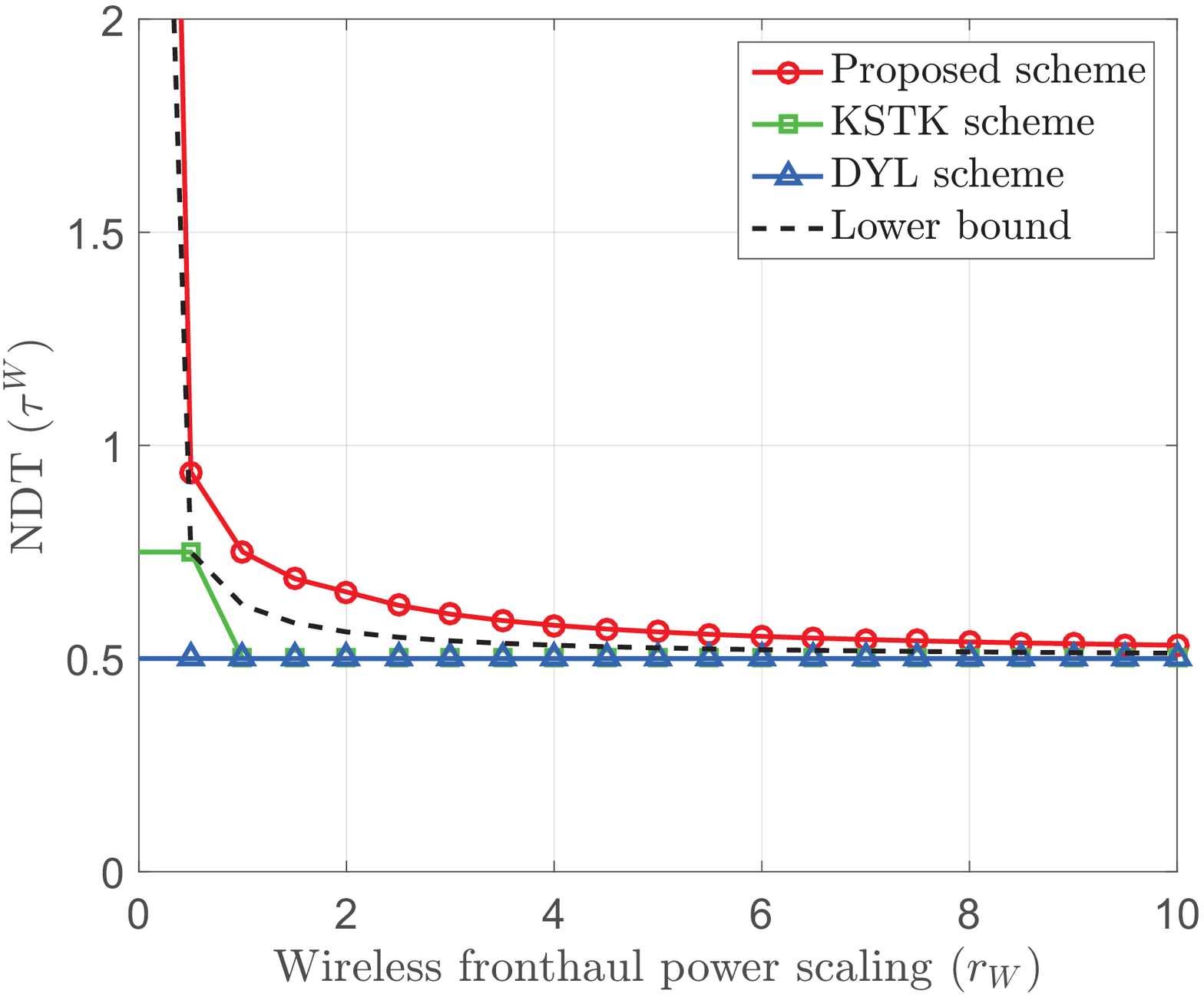}
\subcaption{}\label{Fig compare2 mut033}
\end{minipage}
\caption{NDT v.s. fronthaul capacity $r_W$ when $K_T=2,K_R=2,\mu_R=\frac{1}{2}$, a) $\mu_T=0$ b) $\mu_T=1/2$.}\label{Fig compare2}
\end{figure}

In Fig. \ref{Fig compare2}, we illustrate the performance comparison in a $2\times 2$ F-RAN, to be able to compare with the KSTK scheme, at  $\mu_R=\frac{1}{2}$ and $\mu_R\in \{0, \frac{1}{2}\}$ .
Note that we have added caching at UEs in the plot of the KSTK scheme for fair comparison. We see that when there is no EN cache ($\mu_T=0$) in Fig. \ref{Fig compare2 mut0}, the NDT in our proposed scheme with decentralized caching is smaller than the one in the DYL scheme with centralized caching. This is because our scheme exploits coded multicasting gain in the fronthaul link as stated in Section \ref{section delivery33 m0}, while the DYL scheme  only transmits the uncoded requested file of each UE in the fronthaul link in DYL scheme. The KSTK scheme, on the other hand, has the smallest NDT among the three schemes, and is smaller than our derived lower bound when $r_W>1$, mainly due to the use of full-duplex EN transmission and centralized caching. In specific, the NDT in the KSTK scheme remains $\frac{1}{2}$ when $r_W\ge2$, because the access transmission now becomes the bottleneck in the delivery phase.

When $\mu_T=\frac{1}{2}$ as plotted in  Fig. \ref{Fig compare2 mut033} , it can be seen that the DYL scheme is the best among all in the entire region of $r_W$, in contrast to the finding in Fig. \ref{Fig compare2 mut0}.  The reason is that at this EN cache size, the DYL scheme creates interference-free transmissions for all UEs in the access link by exploiting IC and ZF gains jointly. In fact, by comparing to a simple cut-set lower bound $\tau^W\ge1-\mu_R=\frac{1}{2}$ with centralized caching,  the achievable scheme in DYL is found to be optimal.
 It is further seen from the figure that the KSTK scheme, by full-duplex transmission, is also optimal when $r_W\ge1$, but is sub-optimal when $r_W<1$ because coded multicasting gain based on UE caches in the access link is not exploited. Nevertheless, both DYL and KSTK schemes achieve smaller NDTs than our derived lower bound for decentralized caching. Due to decentralized caching, our scheme is inferior to the other two schemes, but is very close to them when $r_W$ is large, and close to the lower bound in the entire region of $r_W$.

\begin{figure}[!tbp]
\begin{minipage}[t]{1\linewidth}
\centering
\includegraphics[scale=0.4]{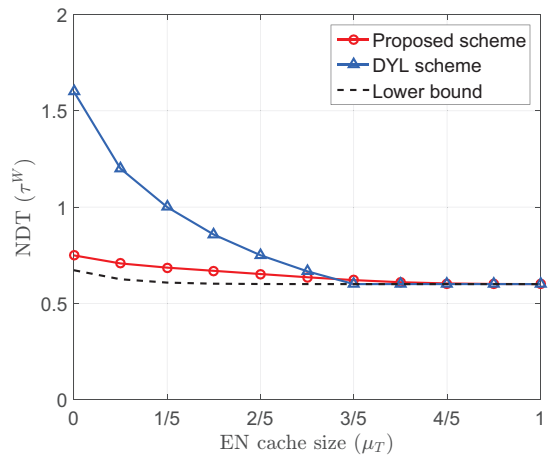}
\subcaption{}\label{Fig compare1 mur04}
\end{minipage}
\begin{minipage}[t]{1\linewidth}
\centering
\includegraphics[scale=0.4]{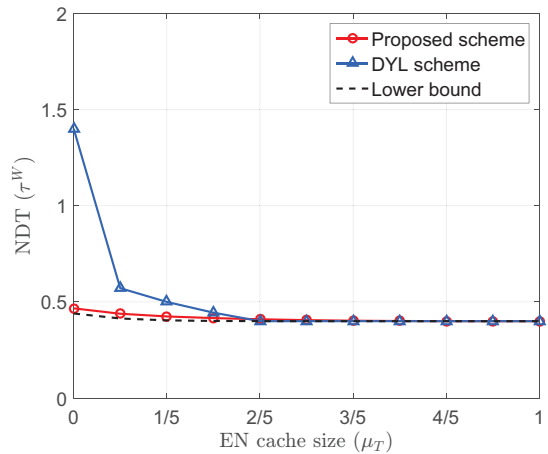}
\subcaption{}\label{Fig compare1 mur06}
\end{minipage}
\caption{NDT v.s. EN cache size $\mu_T$ when $K_T=10,K_R=10,r_W=10$, a) $\mu_R=\frac{2}{5}$ b) $\mu_R=\frac{3}{5}$.}\label{Fig compare1}
\end{figure}

Finally, we compare our scheme with the DYL scheme, with respect to EN cache size $\mu_T$ , in a $10\times 10$ F-RAN with $r_W=10$ at $\mu_R \in \{\frac{2}{5}, \frac{3}{5} \}$ in Fig. \ref{Fig compare1}. Our achievable NDT is still very close to the lower bound. It is even smaller than the one achieved by the DYL scheme with centralized caching when $\mu_T$ is small, and they perform very close when $\mu_T$ is large. This is because our scheme exploits  an additional layer of coded multicasting opportunities in the fronthaul link, while the DYL scheme only exploits the coded multicasting gain in the fronthaul link by generating coded messages directly from requested subfiles $\{W_{q,\Phi,\Psi}\}$. Furthermore, we obtain a larger achievable per-user DoF \eqref{eqn tau ip} than the one in DYL scheme in the access link by using ZF and IA jointly.

\section{Conclusions and Future Directions}\label{section conclusion}

In this paper, we have provided a comprehensive study of content caching and delivery for cache-aided RANs with caches at both the EN and UE sides. We first reviewed the basic tools and techniques that allow exploiting the distributed storage resources in the most efficient manner for cache-aided RANs without fronthaul. Then we have introduced novel transmission techniques for cache-aided RANs with both dedicated and wireless fronthaul links.

%
%

\par
In Section \ref{s:Dedicated}, we have studied the F-RAN architecture with dedicated fronthaul links and an arbitrary number of ENs and UEs, in which both the ENs and the UEs have cache capabilities. We have considered centralized placement at the EN caches, as the ENs represent static access points with dedicated fronthaul links (e.g., millimeter wave connections). We have analyzed both centralized and decentralized cache placement at the UEs. The proposed caching and delivery schemes combine IA, ZF, and IC techniques together with soft-transfer fronthauling, and we have provided comparisons between the achievable NDTs and the literature. We have shown that the proposed schemes reduce the end-to-end delay significantly for a wide range of system parameters taking into account the interplay between the EN caches, UE caches, and the fronthaul link capacities.
\par
In Section \ref{section Caching gain in F-RAN with wireless fronthaul}, we have proposed a novel transmission technique for the F-RAN architecture with a wireless fronthaul link and an arbitrary number of ENs and UEs both equipped with caches. We have analyzed the latency performance under decentralized cache placement at all cache nodes. In the proposed delivery scheme, the wireless fronthaul is used not only to fetch cache-miss contents but also to fetch contents already cached at the ENs to boost EN cooperation to any desired level in the access link. Joint IC, ZF, and IA gains are exploited across the fronthaul and access networks. We have also shown that the latency of the proposed scheme is within a constant multiplicative gap to the optimal. Numerical results show that the proposed delivery scheme with decentralized cache placement can even outperform existing schemes with centralized cache placement under certain conditions.
\par
Among open problems,  to narrow the gap between the achievable schemes and the theoretical bounds, either by further advancing the achievable caching and delivery techniques or by deriving tighter converse, seems to be a natural direction.  Furthermore, the following practical issues are worthwhile for further investigation.
First, given the overhead and feasibility of estimating the global channel state information (CSI), it is of great importance to analyze the caching gain in the presence of imperfect or delayed channel state information. Second, most current caching schemes as well as those considered in this paper still suffer from the exponential sub-packetization problem, which would impede practical implementation when the numbers of ENs and UEs become large \cite{yanqifaPDA};  therefore, low sub-packetization for cache placement is an important research direction \cite{Ozfatura:mobility}. Third, while the new contributions in this paper and a majority of the existing works in the literature focus on the asymptotic NDT analysis, as it lends itself to closed-form expressions, whose order optimality can be proven in certain cases, both finite SNR analysis and finite block-length analysis are also important to validate the conclusions reached through NDT analysis. Low-complexity solutions should also be identified to provide reasonable performance in practical settings. Initial works in these directions can be found in \cite{Mohajer:SPAWC:18, Tolli:ISIT:18}. Last but not least, while this work focuses on the fully connected RAN architecture, it is also important to investigate the more practical partially connected networks, in which each UE can only communicate with a subset of the ENs and/or not all ENs have fronthaul connections. Initial results for partially connected RANs can be found in   \cite{Roushdy:WCNC:18,XuPartial}.

\setcounter{subsubsection}{0}
\section*{Appendix A: Proof of Theorem \ref{thm 1}}
We assume that UE $q$, for $q\in[K_R]$, desires $W_q$ in the delivery phase. Excluding the locally cached subfiles,  each UE $q$, for $q\in[K_R]$, wants subfiles $\{W_{q,\Phi,\Psi}:\Phi\not\ni q,\Phi\subseteq[K_R],\Psi\subseteq[K_T]\}$. We divide the subfiles wanted by all UEs into different groups according to the size of $\Phi$ and $\Psi$, indexed by $\{(m,n):m\in [0:K_R-1], n\in [0:K_T]\}$, such that subfiles in group $(m,n)$ are cached at $m$ UEs and $n$ ENs. There are $K_R\binom{K_R-1}{m}\binom{K_T}{n}$ subfiles in group $(m,n)$, each with fractional size $f_{m,n}$. Each group of subfiles is delivered individually in the time division manner. Without loss of generality, we present the delivery strategy for an arbitrary group $(m,n)$.  The delivery strategy is also given in Algorithm \ref{algorithm 1}.

\begin{algorithm}[!t]
\caption{Delivery scheme for $K_T\times K_R$ F-RAN with wireless fronthaul}\label{algorithm 1}
\begin{algorithmic}[1]
\FOR{$m=0,1,\ldots,K_R-1$}
\FOR{$n=0,1,\ldots,K_T$}
\IF{$n=0$}\STATE Generate coded messages $\{W_{\Phi^+,\emptyset}^{\oplus}\triangleq\bigoplus\limits_{q\in\Phi^+} W_{q,\Phi^+\backslash\{q\},\emptyset}:\Phi^+\subseteq[K_R],|\Phi^+|=m+1\}$, each desired by $m+1$ UEs
\STATE The MBS sends messages  $\{W_{\Phi^+,\emptyset}^\oplus\}$ to all the $K_T$ ENs one by one
\STATE The network topology in the access link is changed into the $\binom{K_T}{K_T}\times\binom{K_R}{m+1}$ cooperative X-multicast channel whose achievable per-user DoF is $d_{m,K_T}$ in \eqref{eqn tau ip}\ELSE
\STATE Generate coded messages $\{W_{\Phi^+,\Psi}^{\oplus}\triangleq\bigoplus\limits_{q\in\Phi^+} W_{q,\Phi^+\backslash\{q\},\Psi}:\Phi^+\subseteq[K_R],|\Phi^+|=m+1,\Psi\subseteq[K_T],|\Psi|=n\}$
\STATE Let $i=\arg\min_{i}  \tau_{m,n}^i$ in \eqref{eqn tau mn gamma}
\STATE Split each coded message into $\binom{K_T-n}{i}$ sub-messages $\{W_{\Phi^+,\Psi}^{\oplus,\Psi^+}\}$, each with fractional size $\frac{f_{m,n}}{\binom{K_T-n}{i}}$ and corresponding to a unique EN set $\Psi^+:|\Psi^+|=n+i, \Psi\subseteq\Psi^+$
\FOR{$\Psi^+\subseteq[K_T],|\Psi^+|=n+i$}
\FOR{$\Phi^+\subseteq[K_R],|\Phi^+|=m+1$}
\IF{$1\le\frac{i}{n+1}$}\STATE The MBS sends sub-messages  $\{W_{\Phi^+,\Psi}^{\oplus,\Psi^+}:\Psi\subseteq\Psi^+,|\Psi|=n\}$ to EN set $\Psi^+$ one by one\ELSE
\STATE The MBS sends coded sub-messages  $\{\bigoplus_{\Psi\subset\Psi'} W_{\Phi^+,\Psi}^{\oplus,\Psi^+}:\Psi'\subseteq\Psi^+,|\Psi'|=n+1,|\Psi|=n\}$ to EN set $\Psi^+$\ENDIF
\STATE ENs in $\Psi^+$ can access $\{W_{\Phi^+,\Psi}^{\oplus,\Psi^+}:\Psi\subseteq\Psi^+,|\Psi|=n\}$ desired by UE set $\Phi^+$.
\ENDFOR
\ENDFOR
\STATE The network topology in the access link is changed into the $\binom{K_T}{n+i}\times\binom{K_R}{m+1}$ cooperative X-multicast channel whose achievable per-user DoF is $d_{m,n+i}$ in \eqref{eqn tau ip}
\ENDIF
\ENDFOR
\ENDFOR
\end{algorithmic}
\end{algorithm}

\subsubsection{$n=0$}
Note that each subfile in group $(m,0)$ is desired by one UE, and already cached at $m$ different UEs but none of ENs. IC  approach can be used. In specific, the coded messages are given by
\begin{align}
\left\{W_{\Phi^+,\emptyset}^{\oplus}\triangleq\!\bigoplus_{q\in\Phi^+} W_{q,\Phi^+\backslash\{q\},\emptyset}:\Phi^+\subseteq[K_R],|\Phi^+|=m+1\right\}.\label{eqn m0 message}
\end{align}
Each coded message $W_{\Phi^+,\emptyset}^{\oplus}$ is desired by UE set $\Phi^+$. (If $m=0$, each coded message $W_{\Phi^+,\emptyset}^\oplus$ degenerates to subfile $W_{q,\emptyset,\emptyset}$ for $\Phi^+=\{q\}$.) These messages need to be generated at the MBS and then delivered to UEs via the fronthaul link and the access link. In the fronthaul link, we let the MBS multicast each coded message in \eqref{eqn m0 message} to all the $K_T$ ENs one
by one. The fronthaul NDT is given by
\begin{align}
\tau_F=  \binom{K_R}{m+1}\frac{f_{m,0}}{r_W}.\label{eqn m0 fronthaul}
\end{align}

By such naive multicast transmission in the fronthaul link, each EN now has access to all the coded messages in \eqref{eqn m0 message}, and can cooperatively transmit together in the access link. The access link thus becomes the $\binom{K_T}{K_T}\times\binom{K_R}{m+1}$ cooperative X-multicast channel with achievable per-user DoF $d_{m,K_T}$ in \eqref{eqn tau ip}. Since each UE desires $\binom{K_R-1}{m}$ messages, the access NDT is given by
\begin{align}
 \tau_A= \binom{K_R-1}{m}\frac{f_{m,0}}{d_{m,K_T}}.\label{eqn m0 access}
\end{align}
Combining \eqref{eqn m0 fronthaul} and \eqref{eqn m0 access}, the achievable NDT for the delivery of group $(m,0)$ is
\begin{align}
  \tau_{m,0}=\binom{K_R}{m+1}\frac{f_{m,0}}{r_W}+\binom{K_R-1}{m}\frac{f_{m,0}}{d_{m,K_T}}.\label{eqn m0}
\end{align}

\subsubsection{$n>0$}
Note that each subfile in group $(m,n)$ is desired by one UE, and already cached at $m$ different UEs and $n$ different ENs. IC approach can be used as when $n=0$. In specific, given an arbitrary UE set $\Phi^+$ with size $|\Phi^+|=m+1$ and an arbitrary EN set $\Psi$ with size $n$, each EN in $\Psi$ generates the coded message $W_{\Phi^+,\Psi}^{\oplus}\triangleq\bigoplus_{q\in\Phi^+} W_{q,\Phi^+\backslash\{q\},\Psi}$ desired by all UEs in $\Phi^+$. (If $m=0$, coded message $W_{\Phi^+,\Psi}^\oplus$ degenerates to subfile $W_{q,\emptyset,\Psi}$ for $\Phi^+=\{q\}$.) Through this IC  approach, $m+1$ different subfiles are combined into a single coded message via XOR, and there are only $\binom{K_R}{m+1}\binom{K_T}{n}$ coded messages to be transmitted in total, each available at $n$ ENs and desired by $m+1$ UEs.

With the aid of fronthaul, we can allow ENs to access the coded messages of others via the transmission of the MBS in the fronthaul link, thereby enabling chances for more transmission cooperation in the access link. Assume that after the aid of fronthaul transmission, every set of $n + i$ ENs can form a cooperation group in the access link, where $i \in [0:K_T-n]$ is a design parameter.\footnote{If $i=0$, every set of $n$ ENs already forms a cooperation group in the access link, and the coded messages can be delivered to UEs directly in the access link without the use of fornthaul channel. The access link becomes the $\binom{K_T}{n}\times\binom{K_R}{m+1}$ cooperative X-multicast channel with per-user DoF of $d_{m,n}$ in \eqref{eqn tau ip}.} We split each coded message $W_{\Phi^+,\Psi}^{\oplus}$ into $\binom{K_T-n}{i}$ sub-messages, each with fractional size $f_{m,n}/\binom{K_T-n}{i}$ and corresponding to a distinct EN set $\Psi^+$ with size $n+i$ such that $\Psi\subseteq\Psi^+$. Denote $W_{\Phi^+,\Psi}^{\oplus,\Psi^+}$ as the sub-message in $W_{\Phi^+,\Psi}^{\oplus}$, which is desired by UE set $\Phi^+$, cached at EN set $\Psi$, and corresponding to EN set $\Psi^+$. Each sub-message $W_{\Phi^+,\Psi}^{\oplus,\Psi^+}$ is sent by EN set $\Psi^+$ exclusively in the access link. Then, for an arbitrary EN set $\Psi^+$ with size $n+i$, each EN in $\Psi^+$ needs to access all the sub-messages
\begin{align}
  \left\{W_{\Phi^+,\Psi}^{\oplus,\Psi^+}:\Phi^+\subseteq[K_R],|\Phi^+|=m+1,\Psi\subseteq\Psi^+,|\Psi|=n\right\}.\label{eqn mn 1}
\end{align}

To do this, the MBS choose one of the two methods below to send sub-messages to ENs in the fronthaul link.
\begin{enumerate}
  \item  Fronthaul Transmission without IC : For each EN set $\Psi^+$, the MBS directly sends sub-messages in \eqref{eqn mn 1} one-by-one, and each EN in $\Psi^+$ decodes all the non-cached sub-messages. By this method, the NDT in the fronthaul link is given by
\begin{align}
 \tau_F^1=\lim_{P\to\infty}\lim_{F\to\infty}\frac{T_F}{F/\log P}=\frac{1}{r_W}\binom{K_R}{m+1}\binom{K_T}{n+i}\binom{n+i}{n}\frac{f_{m,n}}{\binom{K_T-n}{i}}.\label{eqn fronthaul 1}
\end{align}
  \item Fronthaul Transmission with IC : Note that each sub-message is already cached at $n$ ENs. The MBS can exploit IC  opportunities in the fronthaul link. In specific, for each EN set $\Psi^+$, the MBS sends coded sub-messages
\begin{align}
  \bigg\{\bigoplus_{\Psi\subset\Psi'} W_{\Phi^+,\Psi}^{\oplus,\Psi^+}:&\Phi^+\subseteq[K_R],|\Phi^+|=m+1,\Psi'\subseteq\Psi^+,|\Psi'|=n+1,|\Psi|=n\bigg\}.\notag
\end{align}
For each coded sub-message $\bigoplus_{\Psi\subset\Psi'} W_{\Phi^+,\Psi}^{\oplus,\Psi^+}$, each EN $p$ in $\Psi'$ caches $n$ sub-messages $\{W_{\Phi^+,\Psi}^{\oplus,\Psi^+}:p\in\Psi,\Psi\subset\Psi'\}$, and can decode the non-cached sub-message $\{W_{\Phi^+,\Psi}^{\oplus,\Psi^+}:p\notin\Psi,\Psi\subset\Psi'\}$.  By this method the NDT in the fronthaul link is given by
\begin{align}
   \tau_F^2=\lim_{P\to\infty}\lim_{F\to\infty}\frac{T_F}{F/\log P}=\frac{1}{r_W}\binom{K_R}{m+1}\binom{K_T}{n+i}\binom{n+i}{n+1}\frac{f_{m,n}}{\binom{K_T-n}{i}}.\label{eqn fronthaul 2}
\end{align}
\end{enumerate}
Choosing the smaller one between \eqref{eqn fronthaul 1} and \eqref{eqn fronthaul 2}, the fronthaul NDT is given by
\begin{align}
  \tau_F  =&\frac{1}{r_W}\binom{K_R}{m\!+\!1}\binom{K_T}{n\!+\!i}\frac{f_{m,n}}{\binom{K_T-n}{i}}\min\left\{\binom{n\!+\!i}{n},\binom{n\!+\!i}{n\!+\!1}\right\}\notag\\
  =&\binom{K_R}{m+1}\binom{K_T}{n}\min\left\{1,\frac{i}{n+1}\right\}\frac{f_{m,n}}{r_W}.\label{eqn fronthaul ndt}
\end{align}

Then in the access link, for an arbitrary EN set $\Psi^+$ with size $n+i$, each EN in $\Psi^+$ cooperatively sends sub-messages in \eqref{eqn mn 1}. The access link is changed to the $\binom{K_T}{n+i}\times\binom{K_R}{m+1}$ cooperative X-multicast channel with achievable per-user DoF $d_{m,n+i}$ in \eqref{eqn tau ip}. Since each UE $q$, for $q\in[K_R]$, wants $\binom{K_R-1}{m}\binom{K_T}{n+i}\binom{n+i}{n}$ sub-messages, the access NDT is
\begin{align}
  \tau_A=\lim_{P\to\infty}\lim_{F\to\infty}\frac{T_A}{F/\log P}=\binom{K_R-1}{m}\binom{K_T}{n+i}\frac{\binom{n+i}{n}}{\binom{K_T-n}{i}}\frac{f_{m,n}}{d_{m,n+i}}=\binom{K_R-1}{m}\binom{K_T}{n}\frac{f_{m,n}}{d_{m,n+i}}.\label{eqn access ndt}
\end{align}

Combining \eqref{eqn fronthaul ndt} and \eqref{eqn access ndt} and taking the minimum of NDT over $i$, we obtain the NDT for the delivery of group $(m,n)$ as
\begin{align}
\tau_{m,n}=\min_{i \in [0:K_T-n]} \tau_{m,n}^i,\label{eqn mn}
\end{align}
where
\begin{align}
  \tau_{m,n}^i=\binom{K_R}{m+1}\binom{K_T}{n}\min\left\{1,\frac{i}{n+1}\right\}\frac{f_{m,n}}{r_W}+\binom{K_R-1}{m}\binom{K_T}{n}\frac{f_{m,n}}{d_{m,n+i}}.\notag
\end{align}

Summing up NDTs in \eqref{eqn m0} and \eqref{eqn mn} for all groups, the total achievable NDT is
\begin{align}
  \tau^W=\sum_{m=0}^{K_R-1}\sum_{n=0}^{K_T}\tau_{m,n},\notag
\end{align}
which is the same as in Theorem \ref{thm 1}. Thus, Theorem \ref{thm 1} is proved.

\section*{Appendix B:Proof of Theorem \ref{thm 2}}

Since this is a lower bound, we focus on a specific UE demand that each UE $q$ ($q\in[K_R]$) wants file $W_q$. Since ENs are assumed to be half-duplex, we will prove lower bounds on the fronthaul and access NDTs separately.

\subsubsection{Fronthaul Transmission}
We first focus on the fronthaul transmission. Consider the transmission of the files desired by the first $l_1$ UEs, $l_1 \leq K_R$. The proof is based on the following observation. Given received signals $Q^W_{1\sim K_T}$ from the MBS at all ENs and the cache contents $U_{1\sim K_T}$ of all the ENs, one can construct the transmitted signals of all the ENs. Then, given all the transmitted signals from the ENs and cache contents $V_{1\sim l_1}$ at the first $l_1$ UEs, one can recover the desired files of these UEs with arbitrarily low probability of error. We have
\begin{align}
H(W_{1\sim l_1}|Q^W_{1\sim K_T},U_{1\sim K_T},V_{1\sim l_1})=F\varepsilon_F+T_F\varepsilon_P\log P.\notag
\end{align}
Here, $\varepsilon_F$ is a function of file size $F$, and $\varepsilon_P$ is a function of power $P$, and satisfy $\lim_{F\to\infty}\varepsilon_F=0$, $\lim_{P\to\infty}\varepsilon_P=0$. Then, we have
\begin{subequations}\label{eqn converse 1}
\begin{align}
  l_1F=&H(W_{1\sim l_1}|W_{(l_1+1)\sim N})\label{eqn converse 11}\\
  =&I(W_{1\sim l_1};Q^W_{1\sim K_T},U_{1\sim K_T},V_{1\sim l_1}|W_{(l_1+1)\sim N}+H(W_{1\sim l_1}|Q^W_{1\sim K_T},U_{1\sim K_T},V_{1\sim l_1},W_{(l_1+1)\sim N})\label{eqn converse 12}\\
  =&h(Q^W_{1\sim K_T},U_{1\sim K_T},V_{1\sim l_1}|W_{(l_1+1)\sim N})-h(Q^W_{1\sim K_T},U_{1\sim K_T},V_{1\sim l_1}|W_{1\sim N})+F\varepsilon_F+T_F\varepsilon_P\log P\label{eqn converse 13}\\
  \le& h(Q^W_{1\sim K_T})+H(U_{1\sim K_T},V_{1\sim l_1}|W_{(l_1+1)\sim N})+F\varepsilon_F+T_F\varepsilon_P\log P.\label{eqn converse 15}
\end{align}
\end{subequations}
Here, \eqref{eqn converse 12} and \eqref{eqn converse 13} simply follow from the definition of mutual information; and \eqref{eqn converse 15} from the fact that conditioning reduces entropy. In \eqref{eqn converse 15}, $h(Q^W_{1\sim K_T})$ can be bounded by
\begin{subequations}\label{eqn converse 2}
\begin{align}
  h(Q^W_{1\sim K_T})=&I(Q^W_{1\sim K_T};S)+h(Q^W_{1\sim K_T}|S)\label{eqn converse 21}\\
  =&I(Q^W_{1\sim K_T};S)+T_F\varepsilon_P\log P\label{eqn converse 22}\\
  \le&T_F(r_W\log P+\varepsilon_P\log P)+T_F\varepsilon_P\log P,\label{eqn converse 23}
\end{align}
\end{subequations}
where $S$ denotes the signal transmitted by the MBS. Here, \eqref{eqn converse 22} is due to the fact that the conditional entropy $h(Q^W_{1\sim K_T}|S)$ reduces to the entropy of the noise term; and \eqref{eqn converse 23} follows from the capacity bound on the broadcast channel in the high SNR regime.

By denoting $U_{1\sim K_T,1\sim l_1}$ and $V_{1\sim l_1,1\sim l_1}$ as the cached contents of files $W_{1\sim l_1}$ at all the $K_T$ ENs and UEs $\{1,2,\ldots, l_1\}$, respectively, and denoting $U_{1\sim K_T,n},V_{1\sim l_1,n}$ as the cached contents of file $W_n$ at all the $K_T$ ENs and UEs $\{1,2,\ldots, l_1\}$, respectively, the second term in \eqref{eqn converse 15} is given by
\begin{subequations}\label{eqn converse 3}
\begin{align}
  H(U_{1\sim K_T},V_{1\sim l_1}|W_{(l_1+1)\sim N}) =&H(U_{1\sim K_T,1\sim l_1},V_{1\sim l_1,1\sim l_1})\label{eqn converse 31}\\
  =&\sum_{n=1}^{l_1} H(U_{1\sim K_T,n},V_{1\sim l_1,n})\label{eqn converse 32}\\
  =&l_1F\cdot [1-(1-\mu_T)^{K_T}(1-\mu_R)^{l_1}],\label{eqn converse 33}
\end{align}
\end{subequations}
where \eqref{eqn converse 31} and \eqref{eqn converse 32} follow from the fact that only the cached contents of files $\{W_1,\ldots,W_{l_1}\}$ are unknown given files $\{W_{l_1+1},\ldots,W_N\}$ and that the caching scheme does not allow intra-file or inter-file coding; \eqref{eqn converse 33} follows from the fact that each EN and each UE caches a subset of $\mu_TF$ and $\mu_RF$ bits of each file independently and uniformly at random, respectively.

Combining \eqref{eqn converse 15}\eqref{eqn converse 23}\eqref{eqn converse 33}, and letting $F\rightarrow\infty$, $P\rightarrow\infty$, we obtain that
\begin{align}
  \lim_{P\rightarrow\infty}\lim_{F\rightarrow\infty}\frac{T_F\log P}{F}\ge \frac{1}{r_W}l_1(1-\mu_T)^{K_T}(1-\mu_R)^{l_1}.\label{eqn converse 4}
\end{align}

\subsubsection{Access Transmission}
Next we consider the access transmission. The proof method is an extension of the approach in \cite[Section VI]{niesen} by taking decentralized cache scheme into account. Consider the first $l_2$ UEs, for $l_2\in[K_R]$. The proof is based on the following observation. Given the received signals $Y_{1\sim l_2}$ and the cached contents $V_{1\sim l_2}$ of the $l_2$ UEs, one can successfully decode the desired files of these $l_2$ UEs. Thus, we have
\begin{align}
  H(W_{1\sim l_2}|Y_{1\sim l_2},V_{1\sim l_2})=F\varepsilon_F.\notag
\end{align}
Similar to \eqref{eqn converse 1}, we have
\begin{subequations}\label{eqn converse 5}
\begin{align}
  l_2F=&H(W_{1\sim l_2}|W_{(l_2+1)\sim N})\label{eqn converse 51}\\
  =&I(W_{1\sim l_2};Y_{1\sim l_2},V_{1\sim l_2}|W_{(l_2+1)\sim N})+H(W_{1\sim l_2}|Y_{1\sim l_2},V_{1\sim l_2},W_{(l_2+1)\sim N})\label{eqn converse 52}\\
  =&h(Y_{1\sim l_2},V_{1\sim l_2}|W_{(l_2+1)\sim N})-h(Y_{1\sim l_2},V_{1\sim l_2}|W_{1\sim N})+F\varepsilon_F\label{eqn converse 54}\\
  \le& h(Y_{1\sim l_2},V_{1\sim l_2}|W_{(l_2+1)\sim N})+F\varepsilon_F\label{eqn converse 55}\\
  \le& h(Y_{1\sim l_2})+H(V_{1\sim l_2}|W_{(l_2+1)\sim N})+F\varepsilon_F.\label{eqn converse 56}
\end{align}
\end{subequations}
In \eqref{eqn converse 56}, $h(Y_{1\sim l_2})$ is bounded by
\begin{subequations}\label{eqn converse 6}
\begin{align}
  h(Y_{1\sim l_2})=&I(Y_{1\sim l_2};X_{1\sim K_T})+h(Y_{1\sim l_2}|X_{1\sim K_T})\label{eqn converse 61}\\
  =&I(Y_{1\sim l_2};X_{1\sim K_T})+T_F\varepsilon_P\log P\label{eqn converse 62}\\
  \le&T_A\min\{K_T,l_2\}(\log P+\varepsilon_P\log P)+T_A\varepsilon_P\log P.\label{eqn converse 63}
\end{align}
\end{subequations}
Here, \eqref{eqn converse 62} is due to the fact that the conditional entropy $h(Y_{1\sim l_2}|X_{1\sim K_T})$ results only from the noise received at UEs; \eqref{eqn converse 63} follows from the capacity bound of the $K_T\times l_2$ MIMO channel in high SNR regime, similar to the proof of \cite[Lemma 5]{niesen}.

In \eqref{eqn converse 56}, $H(V_{1\sim l_2}|W_{(l_2+1)\sim N})$ is given by
\begin{align}
  H(V_{1\sim l_2}|W_{(l_2+1)\sim N})=H(V_{1\sim l_2,1\sim l_2})=\sum_{n=1}^{l_2} H(V_{1\sim l_2,n})=l_2F\cdot [1-(1-\mu_R)^{l_2}].\label{eqn converse 7}
\end{align}
Note that \eqref{eqn converse 7} is similar to \eqref{eqn converse 3}, and the detailed explanation is omitted here.

Combining \eqref{eqn converse 56}, \eqref{eqn converse 63}, \eqref{eqn converse 7}, and letting $F\rightarrow\infty$, $P\rightarrow\infty$, we obtain
\begin{align}
  \lim_{P\rightarrow\infty}\lim_{F\rightarrow\infty}\frac{T_A\log P}{F}\ge \frac{l_2(1-\mu_R)^{l_2}}{\min\{l_2,K_T\}}.\label{eqn converse 8}
\end{align}

Combining \eqref{eqn converse 4} and \eqref{eqn converse 8}, and taking the maximum over $l_1,l_2\in[K_R]$, the minimum NDT $\tau$ is lower bounded by
\begin{align}
  \tau=\lim_{P\rightarrow\infty}\lim_{F\rightarrow\infty}\frac{(T_F+T_A)\log P}{F}\ge\max_{l_1\in[K_R]}\frac{l_1}{r_W}(1-\mu_T)^{K_T}(1-\mu_R)^{l_1}+ \max_{l_2\in[K_R]}\frac{l_2(1-\mu_R)^{l_2}}{\min\{l_2,K_T\}},\notag
\end{align}
which completed the proof of Theorem \ref{thm 2}.

\section*{Appendix C: Proof of Corollary \ref{coro gap}}
We consider two cases separately, $K_T\ge K_R$ and $K_T<K_R$.

\subsection{$K_T\ge K_R$}
When $K_T\ge K_R$, NDT can be upper bounded by $\tau^W_{upper}=\sum_{m=0}^{K_R-1}\sum_{n=0}^{K_T}\tau_{m,n}$, where $\tau_{m,n}$ is given in \eqref{eqn tau m0} and \eqref{eqn tau mn gamma}. Taking $i=0$ in \eqref{eqn tau mn gamma}, $\tau_{m,n}$ ($n>0$) is bounded by
\begin{align}
  \tau_{m,n}\le\frac{\binom{K_R-1}{m}\binom{K_T}{n}f_{m,n}}{d_{m,n}}.\notag
\end{align}
We also have
\begin{align}
  \tau_{m,0} = \binom{K_R}{m+1}\frac{f_{m,0}}{r} + \frac{\binom{K_R-1}{m}f_{m,0}}{d_{m,K_T}}.\notag
\end{align}
When $K_T\ge K_R$, it is easy to see that $d_{m,n}\ge1/2$ for $m\in[0:K_R-1],n\in[K_T]$. Then, $\tau^W_{upper}$ is upper bounded by
\begin{align}
\tau^W_{upper}\le&\sum_{m=0}^{K_R-1}\sum_{n=1}^{K_T}\frac{\binom{K_R-1}{m}\binom{K_T}{n}f_{m,n}}{d_{m,n}}+\sum_{m=0}^{K_R-1}\frac{\binom{K_R-1}{m}f_{m,0}}{d_{m,K_T}}+\frac{1}{r}\sum_{m=0}^{K_R-1}\binom{K_R}{m+1}f_{m,0}\notag\\
\le&2\sum_{m=0}^{K_R-1}\sum_{n=0}^{K_T}\binom{K_R-1}{m}\binom{K_T}{n}f_{m,n}+\frac{(1-\mu_T)^{K_T}}{r}\sum_{m=0}^{K_R-1}\binom{K_R}{m+1}\mu_R^m(1-\mu_R)^{K_R-m}\notag\\
=&2(1-\mu_R)\sum_{m=0}^{K_R-1}\sum_{n=0}^{K_T}\binom{K_R-1}{m}\binom{K_T}{n}\mu_R^m(1-\mu_R)^{K_R-1-m}\mu_T^n(1-\mu_T)^{K_T-n}\notag\\
&+\frac{(1-\mu_T)^{K_T}}{r}\frac{1-\mu_R}{\mu_R}\sum_{m=0}^{K_R-1}\binom{K_R}{m+1}\mu_R^{m+1}(1-\mu_R)^{K_R-m-1}\notag\\
=&2(1-\mu_R)+\frac{(1-\mu_T)^{K_T}}{r}\frac{1-\mu_R}{\mu_R}\left[\sum_{p=0}^{K_R}\binom{K_R}{p}\mu_R^{p}(1-\mu_R)^{K_R-p}-(1-\mu_R)^{K_R}\right]\notag\\
=&2(1-\mu_R)+\frac{(1-\mu_T)^{K_T}}{r}\frac{1-\mu_R}{\mu_R}\left[1-(1-\mu_R)^{K_R}\right]\label{eqn gap 1}
\end{align}
Taking $l_2=1$ in \eqref{eqn thm 2}, the lower bound of NDT is lower bounded by
\begin{align}
  \tau^W_{lower}\ge\max_{l_1}\frac{l_1(1-\mu_T)^{K_T}(1-\mu_R)^{l_1}}{r}+(1-\mu_R).\label{eqn gap 2}
\end{align}
Denote $g$ as the multiplicative gap, then the gap is bounded by
\begin{align}
  g\le\frac{2(1-\mu_R)+\frac{(1-\mu_T)^{K_T}}{r}\frac{1-\mu_R}{\mu_R}\left[1-(1-\mu_R)^{K_R}\right]}{\max_{l_1}\frac{l_1(1-\mu_T)^{K_T}(1-\mu_R)^{l_1}}{r}+(1-\mu_R)}.\notag
\end{align}
To upper bound $g$, we first consider
\begin{align}
  g_F\triangleq\frac{\frac{(1-\mu_T)^{K_T}}{r}\frac{1-\mu_R}{\mu_R}\left[1-(1-\mu_R)^{K_R}\right]}{\max_{l_1}\frac{l_1(1-\mu_T)^{K_T}(1-\mu_R)^{l_1}}{r}}=\frac{\frac{1-\mu_R}{\mu_R}\left[1-(1-\mu_R)^{K_R}\right]}{\max_{l_1}l_1(1-\mu_R)^{l_1}},\notag
\end{align}
which can also be viewed as the gap in the fronthaul link. We consider four cases to upper bound $g_F$, i.e., (1) $K_R\le 12$; (2) $K_R\ge 13, \mu_R\ge\frac{1}{12}$; (3) $K_R\ge 13, \frac{1}{K_R}\le\mu_R<\frac{1}{12}$; (4) $K_R\ge13, \mu_R<\frac{1}{K_R}$. Note that the broadcast channel in the fronthaul link is similar to the one-server shared link in \cite{fundamentallimits,decentralized}, and the proof here is similar to the one in  \cite{fundamentallimits,decentralized}.
\subsubsection{$K_R\le 12$}
In this case, using the inequality $(1-\mu_R)^{K_R}\ge1-K_R\mu_R$, we have
\begin{align}
\frac{1-\mu_R}{\mu_R}\left[1-(1-\mu_R)^{K_R}\right]\le\frac{1-\mu_R}{\mu_R}K_R\mu_R\le12(1-\mu_R).\notag
\end{align}
Letting $l_1=1$, $g_F$ is bounded by
\begin{align}
  g_F\le\frac{12(1-\mu_R)}{1-\mu_R}=12.\notag
\end{align}
\subsubsection{$K_R\ge 13, \mu_R\ge\frac{1}{12}$}
We have
\begin{align}
\frac{1-\mu_R}{\mu_R}\left[1-(1-\mu_R)^{K_R}\right]\le\frac{1-\mu_R}{\mu_R}\le12(1-\mu_R).\notag
\end{align}
Similar to Case 1 that $K_R\le12$, $g_F$ is also upper bounded by 12.
\subsubsection{$K_R\ge 13, \frac{1}{K_R}\le\mu_R<\frac{1}{12}$}
Letting $l_1=\lfloor\frac{1}{4\mu_R}\rfloor$, we have
\begin{align}
  \max_{l_1}l_1(1-\mu_R)^{l_1}\ge\lfloor\frac{1}{4\mu_R}\rfloor(1-\mu_R)^{\lfloor\frac{1}{4\mu_R}\rfloor}\ge (\frac{1}{4\mu_R}-1)(1-\frac{1}{4\mu_R}\mu_R)=\frac{3}{16\mu_R}-\frac{3}{4}.\notag
\end{align}
Then, $g_F$ is upper bounded by
\begin{align}
  g_F\le\frac{\frac{1-\mu_R}{\mu_R}\left[1-(1-\mu_R)^{K_R}\right]}{\frac{3}{16\mu_R}-\frac{3}{4}}\le\frac{1/\mu_R}{\frac{3}{16\mu_R}-\frac{3}{4}}=\frac{1}{3/16-3\mu_R/4}<\frac{1}{3/16-3/48}=8.\notag
\end{align}
\subsubsection{$K_R\ge13, \mu_R<\frac{1}{K_R}$}
Letting $l_1=\lfloor\frac{K_R}{4}\rfloor$, we have
\begin{align}
g_F&\le\frac{\frac{1-\mu_R}{\mu_R}\left[1-(1-\mu_R)^{K_R}\right]}{\lfloor\frac{K_R}{4}\rfloor(1-\mu_R)^{\lfloor\frac{K_R}{4}\rfloor}}\notag\\
&=\frac{1-(1-\mu_R)^{K_R}}{\mu_R\lfloor\frac{K_R}{4}\rfloor(1-\mu_R)^{\lfloor\frac{K_R}{4}\rfloor-1}}\notag\\
&\le\frac{1-(1-K_R\mu_R)}{\mu_R\lfloor\frac{K_R}{4}\rfloor(1-\mu_R)^{\lfloor\frac{K_R}{4}\rfloor-1}}\notag\\
&=\frac{K_R}{\lfloor\frac{K_R}{4}\rfloor}\frac{1}{(1-\mu_R)^{\lfloor\frac{K_R}{4}\rfloor-1}}\notag\\
&\le\frac{K_R}{\frac{K_R}{4}-1}\frac{1}{1-(\frac{K_R}{4}-1)\mu_R}\notag\\
&<\frac{1}{\frac{1}{4}-\frac{1}{K_R}}\frac{1}{1-(\frac{K_R}{4}-1)\frac{1}{K_R}}\notag\\
&\le\frac{1}{\frac{1}{4}-\frac{1}{13}}\frac{1}{\frac{3}{4}+\frac{1}{K_R}}<8.\notag
\end{align}
Combining all four cases, we find that $g_F\le12$ for all $\mu_R,K_R$. Then, the gap $g$ is upper bounded by
\begin{align}
  g\le\frac{2(1-\mu_R)+12\max_{l_1}\frac{l_1(1-\mu_T)^{K_T}(1-\mu_R)^{l_1}}{r}}{\max_{l_1}\frac{l_1(1-\mu_T)^{K_T}(1-\mu_R)^{l_1}}{r}+(1-\mu_R)}\le12.\notag
\end{align}
Thus, we proved the case when $K_T\ge K_R$.

\subsection{$K_T<K_R$}
Now, we consider the case when $K_T<K_R$. The achievable upper bound of NDT is $\tau^W_{upper}=\sum_{m=0}^{K_R-1}\sum_{n=0}^{K_T}\tau_{m,n}$, where $\tau_{m,n}$ is given in \eqref{eqn tau m0} and \eqref{eqn tau mn gamma}. Taking $i=0$ in \eqref{eqn tau mn gamma}, $\tau_{m,n}$ ($n>0$) is bounded by
\begin{align}
  \tau_{m,n}\le\frac{\binom{K_R-1}{m}\binom{K_T}{n}f_{m,n}}{d_{m,n}}.\notag
\end{align}
We also have
\begin{align}
  \tau_{m,0} = \binom{K_R}{m+1}\frac{f_{m,0}}{r} + \frac{\binom{K_R-1}{m}f_{m,0}}{d_{m,K_T}}.\notag
\end{align}
It is easy to see in \eqref{eqn tau ip} that $d_{m,n}\ge d_{m,1}=\frac{K_T}{K_T+\frac{K_R-m-1}{m+1}}$ for $m\in[0:K_R-1],n\in[K_T]$. Then, the achievable upper bound of NDT is bounded by \eqref{eqn gap 3}.
\begin{figure*}[!t]
\begin{align}
  \tau_{upper}\le&\sum_{m=0}^{K_R-1}\sum_{n=1}^{K_T}\frac{\binom{K_R-1}{m}\binom{K_T}{n}f_{m,n}}{d_{m,n}}+\sum_{m=0}^{K_R-1}\frac{\binom{K_R-1}{m}f_{m,0}}{d_{m,K_T}}+\frac{1}{r}\sum_{m=0}^{K_R-1}\binom{K_R}{m+1}f_{m,0}\notag\\
  \le&\sum_{m=0}^{K_R-1}\sum_{n=1}^{K_T}\frac{\binom{K_R-1}{m}\binom{K_T}{n}f_{m,n}}{\frac{K_T}{K_T+\frac{K_R-m-1}{m+1}}}+\sum_{m=0}^{K_R-1}\frac{\binom{K_R-1}{m}f_{m,0}}{\frac{K_T}{K_T+\frac{K_R-m-1}{m+1}}}+\frac{1}{r}\sum_{m=0}^{K_R-1}\binom{K_R}{m+1}f_{m,0}\notag\\
  =&\sum_{m=0}^{K_R-1}\frac{\binom{K_R-1}{m}}{\frac{K_T}{K_T+\frac{K_R-m-1}{m+1}}}\sum_{n=0}^{K_T}\binom{K_T}{n}f_{m,n}+\frac{1}{r}\sum_{m=0}^{K_R-1}\binom{K_R}{m+1}f_{m,0}\notag\\
  =&\sum_{m=0}^{K_R-1}\frac{\binom{K_R-1}{m}}{\frac{K_T}{K_T+\frac{K_R-m-1}{m+1}}}\mu_R^m(1-\mu_R)^{K_R-m}+\frac{1}{r}\sum_{m=0}^{K_R-1}\binom{K_R}{m+1}f_{m,0}\notag\\
  =&\sum_{m=0}^{K_R-1}\frac{\binom{K_R-1}{m}(K_T-1+\frac{K_R}{m+1})}{K_T}\mu_R^m(1-\mu_R)^{K_R-m}+\frac{1}{r}\sum_{m=0}^{K_R-1}\binom{K_R}{m+1}f_{m,0}\notag\\
  =&\frac{K_T-1}{K_T}\sum_{m=0}^{K_R-1}\binom{K_R-1}{m}\mu_R^m(1-\mu_R)^{K_R-m}+\frac{1}{K_T}\sum_{m=0}^{K_R-1}\binom{K_R}{m+1}\mu_R^m(1-\mu_R)^{K_R-m}\notag\\
  &+\frac{1}{r}\sum_{m=0}^{K_R-1}\binom{K_R}{m+1}f_{m,0}\notag\\
  =&\frac{K_T-1}{K_T}(1-\mu_R)\sum_{m=0}^{K_R-1}\binom{K_R-1}{m}\mu_R^m(1-\mu_R)^{K_R-m-1}+\frac{1}{r}\sum_{m=0}^{K_R-1}\binom{K_R}{m+1}f_{m,0}\notag\\
  &+\frac{1-\mu_R}{K_T\mu_R}\sum_{m=0}^{K_R-1}\binom{K_R}{m+1}\mu_R^{m+1}(1-\mu_R)^{K_R-m-1}\notag\\
  =&\frac{K_T-1}{K_T}(1-\mu_R)+\frac{1-\mu_R}{K_T\mu_R}\sum_{p=1}^{K_R}\binom{K_R}{p}\mu_R^{p}(1-\mu_R)^{K_R-p}+\frac{1}{r}\sum_{m=0}^{K_R-1}\binom{K_R}{m+1}f_{m,0}\notag\\
  =&\frac{K_T-1}{K_T}(1-\mu_R)+\frac{1-\mu_R}{K_T\mu_R}\left[1-(1-\mu_R)^{K_R}\right]+\frac{(1-\mu_T)^{K_T}}{r}\frac{1-\mu_R}{\mu_R}\left[1-(1-\mu_R)^{K_R}\right]\label{eqn gap 3}
\end{align}
\hrule
\end{figure*}
Using Theorem \ref{thm 2}, the multiplicative gap $g$ is bounded by
\begin{align}
  g\le\frac{\frac{K_T-1}{K_T}(1-\mu_R)+\frac{1-\mu_R}{K_T\mu_R}\left[1-(1-\mu_R)^{K_R}\right]+\frac{(1-\mu_T)^{K_T}}{r}\frac{1-\mu_R}{\mu_R}\left[1-(1-\mu_R)^{K_R}\right]}{\max_{l_1\in[K_R]}\frac{l_1(1-\mu_T)^{K_T}(1-\mu_R)^{l_1}}{r}+\max_{l_2\in[K_R]}\frac{l_2(1-\mu_R)^{l_2}}{\min\{l_2,K_T\}}}.\label{eqn gap 4}
\end{align}
In \eqref{eqn gap 4}, from the analysis when $K_T\ge K_R$, we have
\begin{align}
  \frac{\frac{(1-\mu_T)^{K_T}}{r}\frac{1-\mu_R}{\mu_R}\left[1-(1-\mu_R)^{K_R}\right]}{\max_{l_1\in[K_R]}\frac{l_1(1-\mu_T)^{K_T}(1-\mu_R)^{l_1}}{r}}\le12.\notag
\end{align}
Then, to bound $g$ in \eqref{eqn gap 4}, we first consider
\begin{align}
  g_A\triangleq\frac{\frac{K_T-1}{K_T}(1-\mu_R)+\frac{1-\mu_R}{K_T\mu_R}\left[1-(1-\mu_R)^{K_R}\right]}{\max_{l_2\in[K_R]}\frac{l_2(1-\mu_R)^{l_2}}{\min\{l_2,K_T\}}},\notag
\end{align}
which can also be viewed as the multiplicative gap in the access link. We use three cases to upper bound $g_A$, i.e., (1) $\mu_R<\frac{1}{4K_R}$; (2) $\frac{1}{4K_R}\le\mu_R<\frac{1}{4K_T}$; (3) $\mu_R\ge\frac{1}{4K_T}$.
\subsubsection{$\mu_R<\frac{1}{4K_R}$}
Letting $l_2=K_R$, we have
\begin{align}
\max_{l_2\in[K_R]}\frac{l_2(1-\mu_R)^{l_2}}{\min\{l_2,K_T\}}\ge\frac{K_R(1-\mu_R)^{K_R}}{K_T}\ge\frac{K_R(1-K_R\mu_R)}{K_T}>\frac{K_R}{K_T}(1-K_R\frac{1}{4K_R})=\frac{3K_R}{4K_T}.\label{eqn gap 5}
\end{align}
Letting $l_2=1$, we have
\begin{align}
  \max_{l_2\in[K_R]}\frac{l_2(1-\mu_R)^{l_2}}{\min\{l_2,K_T\}}\ge1-\mu_R.\label{eqn gap 6}
\end{align}
We also have
\begin{align}
  \frac{1-\mu_R}{K_T\mu_R}\left[1-(1-\mu_R)^{K_R}\right]\le\frac{1-\mu_R}{K_T\mu_R}\left[1-(1-K_R\mu_R)\right]=\frac{K_R(1-\mu_R)}{K_T}\le\frac{K_R}{K_T}.\label{eqn gap 7}
\end{align}
Combining \eqref{eqn gap 5}\eqref{eqn gap 6}\eqref{eqn gap 7}, $g_A$ is upper bounded by
\begin{align}
  g_A=\frac{\frac{K_T-1}{K_T}(1-\mu_R)+\frac{1-\mu_R}{K_T\mu_R}\left[1-(1-\mu_R)^{K_R}\right]}{\max_{l_2} \frac{l_2(1-\mu_R)^{l_2}}{\min\{l_2,K_T\}}}\le\frac{\frac{K_T-1}{K_T}(1-\mu_R)}{1-\mu_R}+\frac{\frac{K_R}{K_T}}{\frac{3K_R}{4K_T}}<1+4/3=7/3.\notag
\end{align}

\subsubsection{$\frac{1}{4K_R}\le\mu_R<\frac{1}{4K_T}$}
Letting $l_2=\lceil\frac{1}{4\mu_R}\rceil$, we have
\begin{align}
\max_{l_2} \frac{l_2(1-\mu_R)^{l_2}}{\min\{l_2,K_T\}}\ge&\frac{\lceil\frac{1}{4\mu_R}\rceil(1-\mu_R)^{\lceil\frac{1}{4\mu_R}\rceil}}{\min\{\lceil\frac{1}{4\mu_R}\rceil,K_T\}}\notag\\
  \ge&\frac{\frac{1}{4\mu_R}(1-\lceil\frac{1}{4\mu_R}\rceil\mu_R)}{K_T}\notag\\
  \ge&\frac{1-(\frac{1}{4\mu_R}+1)\mu_R}{4\mu_RK_T}\notag\\
  =&\frac{\frac{3}{4}-\mu_R}{4K_T\mu_R}\notag\\
  >&\frac{\frac{3}{4}-\frac{1}{8}}{4K_T\mu_R}=\frac{5}{32K_T\mu_R}.\label{eqn gap 8}
\end{align}
We also have
\begin{align}
  \frac{1-\mu_R}{K_T\mu_R}\left[1-(1-\mu_R)^{K_R}\right]\le\frac{1}{K_T\mu_R}.\label{eqn gap 9}
\end{align}
Combining \eqref{eqn gap 6}\eqref{eqn gap 8}\eqref{eqn gap 9}, $g_A$ is upper bounded by
\begin{align}
  g_A=\frac{\frac{K_T-1}{K_T}(1-\mu_R)+\frac{1-\mu_R}{K_T\mu_R}\left[1-(1-\mu_R)^{K_R}\right]}{\max_{l_2} \frac{l_2(1-\mu_R)^{l_2}}{\min\{l_2,K_T\}}}\le\frac{\frac{K_T-1}{K_T}(1-\mu_R)}{1-\mu_R}+\frac{\frac{1}{K_T\mu_R}}{\frac{5}{32K_T\mu_R}}<1+32/5=37/5.\notag
\end{align}
\subsubsection{$\mu_R\ge\frac{1}{4K_T}$}
Letting $l_2=\lfloor\frac{1}{4\mu_R}\rfloor$, we have
\begin{align}
  \max_{l_2} \frac{l_2(1-\mu_R)^{l_2}}{\min\{l_2,K_T\}}\ge\frac{\lfloor\frac{1}{4\mu_R}\rfloor(1-\mu_R)^{\lfloor\frac{1}{4\mu_R}\rfloor}}{\min\{\lfloor\frac{1}{4\mu_R}\rfloor,K_T\}}=(1-\mu_R)^{\lfloor\frac{1}{4\mu_R}\rfloor}\ge1-\frac{1}{4\mu_R}\mu_R=\frac{3}{4}.\label{eqn gap 10}
\end{align}
We also have
\begin{align}
  \frac{1-\mu_R}{K_T\mu_R}\left[1-(1-\mu_R)^{K_R}\right]\le\frac{1}{K_T\mu_R}\le\frac{1}{K_T\frac{1}{4K_T}}=4.\label{eqn gap 11}
\end{align}
Combining \eqref{eqn gap 6}\eqref{eqn gap 10}\eqref{eqn gap 11}, $g_A$ is bounded by
\begin{align}
   g_A=\frac{\frac{K_T-1}{K_T}(1-\mu_R)+\frac{1-\mu_R}{K_T\mu_R}\left[1-(1-\mu_R)^{K_R}\right]}{\max_{l_2} \frac{l_2(1-\mu_R)^{l_2}}{\min\{l_2,K_T\}}}\le\frac{\frac{K_T-1}{K_T}(1-\mu_R)}{1-\mu_R}+\frac{4}{3/4}<1+16/3=19/3.\notag
\end{align}

From the above three cases, we find that $g_A<12$. Then the multiplicative gap $g$ is bounded by
\begin{align}
  g\le&\frac{\frac{K_T-1}{K_T}(1-\mu_R)+\frac{1-\mu_R}{K_T\mu_R}\left[1-(1-\mu_R)^{K_R}\right]+\frac{(1-\mu_T)^{K_T}}{r}\frac{1-\mu_R}{\mu_R}\left[1-(1-\mu_R)^{K_R}\right]}{\max_{l_1\in[K_R]}\frac{l_1(1-\mu_T)^{K_T}(1-\mu_R)^{l_1}}{r}+\max_{l_2\in[K_R]}\frac{l_2(1-\mu_R)^{l_2}}{\min\{l_2,K_T\}}}\notag\\
  <&\frac{12\max_{l_2\in[K_R]}\frac{l_2(1-\mu_R)^{l_2}}{\min\{l_2,K_T\}}+12\max_{l_1\in[K_R]}\frac{l_1(1-\mu_T)^{K_T}(1-\mu_R)^{l_1}}{r}}{\max_{l_1\in[K_R]}\frac{l_1(1-\mu_T)^{K_T}(1-\mu_R)^{l_1}}{r}+\max_{l_2\in[K_R]}\frac{l_2(1-\mu_R)^{l_2}}{\min\{l_2,K_T\}}}=12.\label{eqn gap 12}
\end{align}

Thus we finished the proof of Corollary \ref{coro gap} that the multiplicative gap is within 12.

\bibliographystyle{IEEEtran}
\bibliography{IEEEabrv,journal}

\end{document}